%% file: CoRR.tex
\documentclass[a4paper,fleqn]{techreport}

\usepackage{hyperref}
\usepackage{graphicx}
\usepackage{color}
\usepackage{amsmath}
\usepackage{amssymb}
\usepackage{QED}

  \newcommand{\Acts}[1][A]{\ensuremath{\mathalpha{#1}}}
  \newcommand{\act}[1]{\ensuremath{\mathalpha{#1}}}
    \newcommand{\acta}[1][]{\ensuremath{\act{a_{#1}}}}
    \newcommand{\actb}[1][]{\ensuremath{\act{b_{#1}}}}
  \newcommand{\sact}{\act{\tau}}
  \newcommand{\silent}[1][\tau]{\ensuremath{\mathalpha{#1}}}
  \newcommand{\States}[1][S]{\ensuremath{\mathalpha{#1}}}
  \newcommand{\state}[1]{\ensuremath{\mathalpha{#1}}}
    \newcommand{\states}[1][]{\state{s_{#1}}}
    \newcommand{\statet}[1][]{\state{t_{#1}}}
    \newcommand{\stateu}[1][]{\state{u_{#1}}}
  \newcommand{\statemap}[1][\sigma]{\ensuremath{\mathalpha{#1}}}
  \newcommand{\brelsym}[1][\mathcal{R}]{\ensuremath{#1}}
  \newcommand{\brel}[1][{\brelsym}]{\ensuremath{\mathrel{#1}}}
   \newdimen\boxwdplusemdimen
   \def\arrow#1{{
     \boxwdplusemdimen=1em%
     \setbox0=\hbox{$\scriptstyle#1$}%
     \advance\boxwdplusemdimen by \wd0\relax%
     \ifdim\boxwdplusemdimen<16.11119pt%
       \boxwdplusemdimen=16.11119pt%
     \fi%
     \buildrel{#1}\over%
       {\setbox1=\hbox to \boxwdplusemdimen{\rightarrowfill}%
     \ht1=0.3em\relax\box1}%
   }}
   \makeatletter
   \def\twoheadrightarrowfill{$\m@th\smash-\mkern-7mu%
     \cleaders\hbox{$\mkern-2mu\smash-\mkern-2mu$}\hfill
     \mkern-7mu\mathord\twoheadrightarrow$}
   \makeatother
   \def\darrow#1{{
     \boxwdplusemdimen=1em%
     \setbox0=\hbox{$\scriptstyle#1$}%
     \advance\boxwdplusemdimen by \wd0\relax%
     \ifdim\boxwdplusemdimen<16.11119pt%
       \boxwdplusemdimen=16.11119pt%
     \fi%
     \buildrel{#1}\over%
       {\setbox1=\hbox to \boxwdplusemdimen{\twoheadrightarrowfill}%
     \ht1=0.3em\relax\box1}%
   }}
   \def\plarrow#1{{
     \boxwdplusemdimen=1em%
     \setbox0=\hbox{$\scriptstyle#1$}%
     \advance\boxwdplusemdimen by \wd0\relax%
     \ifdim\boxwdplusemdimen<16.11119pt%
       \boxwdplusemdimen=16.11119pt%
     \fi%
     \buildrel{#1}\over%
       {\setbox1=\hbox to \boxwdplusemdimen{\rightarrowfill}%
     \ht1=0.3em\relax\box1^{\scriptstyle +}}%
   }}
  \newcommand{\step}[1]{\ensuremath{\mathbin{\arrow{#1}}}}
  \newcommand{\stepsym}{\ensuremath{\mathalpha{\rightarrow}}}
  \newcommand{\ssteps}{\ensuremath{\mathbin{\darrow{}}}}
  \newcommand{\sstepssym}{\ensuremath{\mathbin{\twoheadrightarrow}}}
  \newcommand{\plussteps}{\ensuremath{\mathbin{\plarrow{}}}}
  \newcommand{\plusstepssym}{\ensuremath{\mathbin{\rightarrow^+}}}
  \newcommand{\bisim}{%
    \setbox0=\hbox{\kern-.1ex{$\leftrightarrow$}\kern-.1ex}
    \setbox1=\vbox{\hbox{\raise .1ex \box0}\hrule}%
    \ensuremath{\mathrel{\hbox{\kern.1ex\box1\kern.1ex}}}
  }
  \newcommand{\bbisim}{%
    \setbox0=\hbox{\kern-.1ex{$\leftrightarrow$}\kern-.1ex}
    \setbox1=\vbox{\hbox{\raise .1ex \box0}\hrule}%
    \ensuremath{\mathrel{\hbox{\kern.1ex\box1\kern.1ex}_{b}}}
  }
  \newcommand{\wbbisimd}{%
    \setbox0=\hbox{\kern-.1ex{$\leftrightarrow$}\kern-.1ex}
    \setbox1=\vbox{\hbox{\raise .1ex \box0}\hrule}%
    \ensuremath{\mathrel{\hbox{\kern.1ex\box1\kern.1ex}^{\Delta_3}_{b}}}
  }
  \newcommand{\wbbisimdtwo}{%
    \setbox0=\hbox{\kern-.1ex{$\leftrightarrow$}\kern-.1ex}
    \setbox1=\vbox{\hbox{\raise .1ex \box0}\hrule}%
    \ensuremath{\mathrel{\hbox{\kern.1ex\box1\kern.1ex}^{\Delta_2}_{b}}}
  }
  \newcommand{\onestepbbisimd}{%
    \setbox0=\hbox{\kern-.1ex{$\leftrightarrow$}\kern-.1ex}
    \setbox1=\vbox{\hbox{\raise .1ex \box0}\hrule}%
    \ensuremath{\mathrel{\hbox{\kern.1ex\box1\kern.1ex}^{\Delta_1}_{b}}}
  }
  \newcommand{\bbisimd}{%
    \setbox0=\hbox{\kern-.1ex{$\leftrightarrow$}\kern-.1ex}
    \setbox1=\vbox{\hbox{\raise .1ex \box0}\hrule}%
    \ensuremath{\mathrel{\hbox{\kern.1ex\box1\kern.1ex}^{\Delta_4}_{b}}}
  }
  \newcommand{\rbbisimd}{%
    \setbox0=\hbox{\kern-.1ex{$\leftrightarrow$}\kern-.1ex}
    \setbox1=\vbox{\hbox{\raise .1ex \box0}\hrule}%
    \ensuremath{\mathrel{\hbox{\kern.1ex\box1\kern.1ex}^{\Delta}_{b}}}
  }
  \newcommand{\rbbisimdzero}{%
    \setbox0=\hbox{\kern-.1ex{$\leftrightarrow$}\kern-.1ex}
    \setbox1=\vbox{\hbox{\raise .1ex \box0}\hrule}%
    \ensuremath{\mathrel{\hbox{\kern.1ex\box1\kern.1ex}^{\Delta_0}_{b}}}
  }
  \newcommand{\bbisimsym}{\ensuremath{\mathalpha{\bbisim}}}
  \newcommand{\bbisimdsym}{\ensuremath{\mathalpha{\bbisimd}}}
  \newcommand{\onestepbbisimdsym}{\ensuremath{\mathalpha{\onestepbbisimd}}}
  \newcommand{\wbbisimdsym}{\ensuremath{\mathalpha{\wbbisimd}}}
  \newcommand{\rbbisimdsym}{\ensuremath{\mathalpha{\rbbisimd}}}
  \newcommand{\N}{\ensuremath{\mathalpha{\omega}}}
  \newcommand{\bstate}[1]{\state{#1^\flat}}
  \newcommand{\astate}[1]{\state{#1^\sharp}}
  \newcommand{\bbd}{\mbox{BB$^\Delta$}}
  \newcommand{\bbdsc}{\ensuremath{\brel[\hat{\mathcal{R}}]}}
  \newcommand{\bbdscsym}{\ensuremath{\brelsym[\hat{\mathcal{R}}]}}

  \newcommand{\IFF}{iff}
  \newcommand{\all}[1]{\ensuremath{\mathalpha{\forall{#1}}}}
  \newcommand{\is}[1]{\ensuremath{\mathalpha{\exists{#1}}}}
  \newcommand{\Frm}[1][\Phi]{\ensuremath{\mathalpha{#1}}}
  \newcommand{\frm}[1][\varphi]{\ensuremath{\mathalpha{#1}}}
  \newcommand{\true}{\ensuremath{\mathalpha{\top}}}
  \newcommand{\false}{\ensuremath{\mathalpha{\bot}}}
  \newcommand{\compl}{\ensuremath{\mathalpha{\neg}}}
  \newcommand{\meet}{\ensuremath{\mathbin{\wedge}}}
  \newcommand{\Meet}[1][]{\ensuremath{\bigwedge_{#1}}}
  \newcommand{\sat}{\ensuremath{\mathrel{\models}}}
\newcommand{\runtil}[1][\acta]{\ensuremath{\mathbin{#1}}}
\newcommand{\Div}{\ensuremath{\mathalpha{\widehat\Delta}}}
\newcommand{\bbvalidsym}{\ensuremath{\mathalpha{\approx}}}
\newcommand{\bbvalid}{\ensuremath{\mathrel{\bbvalidsym}}}
\newcommand{\requivalidsym}{\ensuremath{\mathalpha{\approx^{\rDiv}}}}
\newcommand{\requivalid}{\ensuremath{\mathrel{\requivalidsym}}}
\newcommand{\funtil}[1][\acta]{\ensuremath{\mathbin{\langle{\hat#1}\rangle}}}
  \newcommand{\until}[1][\acta]{\ensuremath{\mathbin{\langle{#1}\rangle}}}
  \newcommand{\rDiv}{\ensuremath{\mathalpha{\Delta}}}
  \newcommand{\equivalidsym}{\ensuremath{\mathalpha{\approx^{\Div}}}}
  \newcommand{\equivalid}{\ensuremath{\mathrel{\equivalidsym}}}
  \newcommand{\relcomp}{\ensuremath{\mathbin{;}}}
\newcommand{\opt}[1]{\mbox{\tiny\rm(}#1\mbox{\tiny\rm)}} %
\newcommand{\plat}[1]{\raisebox{0pt}[0pt][0pt]{#1}}      %
\newcommand{\C}{\ensuremath{\mathcal{C}}}
\newcommand{\cc}{\equiv_c}
\newcommand{\ccd}{\equiv_c^\Delta}

\newcommand{\ccdsym}{\ensuremath{\mathalpha{\ccd}}}

\newcommand{\JBFrm}{\ensuremath{\Frm_\mathit{jb}}}
\newcommand{\JBEDFrm}{\ensuremath{\Frm^{\rDiv}_\mathit{jb}}}

\newcommand{\UEDFrm}{\ensuremath{\Frm^{\rDiv}_\mathit{u}}}
\newcommand{\JBSDFrm}{\ensuremath{\Frm^{\Div}_\mathit{jb}}}
\newcommand{\equi}{\Lleftarrow\!\!\!\Rrightarrow}

\newcommand{\ofwhicheverystateis}{of which every state is}%

\title{Branching Bisimilarity with Explicit Divergence}

   \author{Rob van Glabbeek\\
            National ICT Australia, Sydney, Australia\\
            School of Computer Science and Engineering,
            University of New South Wales, Sydney, Australia
   \and Bas Luttik\\
            Department of Mathematics and Computer Science,
            Technische Universiteit Eindhoven,
            The Netherlands\\
            CWI, The Netherlands
   \and Nikola Tr\v{c}ka\\
            Department of Mathematics and Computer Science,
            Technische Universiteit Eindhoven,
            The Netherlands}

\begin{document}
\issue{(submission)}
\runninghead{R.J. van Glabbeek, B. Luttik, N. Tr\v{c}ka}{Branching
  bisimilarity with explicit divergence}
\maketitle
\vspace{-1in}
\begin{abstract}
  We consider the relational characterisation of branching
  bisimilarity with explicit divergence. We prove that it is an
  equivalence and that it coincides with the original definition of
  branching bisimilarity with explicit divergence in terms of coloured
  traces. We also establish a correspondence with several variants of
  an action-based modal logic with until- and divergence modalities.
\end{abstract}

  \section{Introduction}\label{sect:introduction}

  Branching bisimilarity was proposed in \cite{GW96}. It is a
  behavioural equivalence on processes that is compatible with a
  notion of abstraction from internal activity, while at the same
  preserving the branching structure of processes in a strong
  sense. We refer the reader to \cite{GW96}, in particular
  to Section 10 therein, for ample motivation of the relevance of
  branching bisimilarity.

  Branching bisimilarity abstracts to a large extent from
  \emph{divergence} (i.e., infinite internal activity). For instance,
  it identifies a process, say $P$, that may perform some internal
  activity after which it returns to its initial state (i.e., $P$ has
  a $\silent$-loop) with a process, say $P'$, that admits the
  same behaviour as $P$ except that it cannot perform the
  internal activity leading to the initial state (i.e., $P'$ is $P$
  without the $\silent$-loop). This means that branching bisimilarity
  is not compatible with any temporal logic featuring an
  \emph{eventually} modality: for any desired state that $P'$ will
  eventually reach,  the mentioned internal activity of $P$ may
  be performed continuously, and thus prevent $P$ from reaching
  this desired state.

  The notion of \emph{branching bisimilarity with explicit divergence}
  ($\bbd$), also proposed in \cite{GW96}, is a suitable refinement of
  branching bisimilarity that is compatible with the well-known
  branching-time temporal logic CTL$^*$ without the nexttime operator
  $X$ (which is known to be incompatible with abstraction from internal
  activity).
  In fact, in \cite{CTLd} we have proved that it is the coarsest
  semantic equivalence on labelled transition systems with silent
  moves that is a congruence for parallel composition (as found in
  process algebras like CCS, CSP or ACP) and only equates processes
  satisfying the same CTL$^*_{-X}$ formulas. It is also the finest 
  equivalence in the \emph{linear time -- branching time spectrum} of
  \cite{Gla93}.

  There are several ways to characterise a behavioural equivalence.
  The original definition of $\bbd$, in terms of \emph{coloured
    traces}, stems from \cite{GW96}. In \cite{Gla93}, $\bbd$ is
  defined in terms of a modal and a relational characterisation, which
  are claimed to coincide with each other and with the original notion
  from \cite{GW96}. Of these three definitions of $\bbd$, the
  relational characterisation from \cite{Gla93} is the most concise
  one, in the sense that it requires the least amount of auxiliary
  concepts.  Moreover, this definition is most in the style of the
  standard definitions of other kinds of bisimulation, found elsewhere
  in the literature. For these reasons, it is tempting to take it as
  standard definition.

  Although it is not hard to establish that the modal characterisation
  from \cite{Gla93} is correct, in the sense that it defines an
  equivalence that coincides with $\bbd$ of \cite{GW96}, it is not at
  all trivial to establish that the same holds for the relational
  characterisation from \cite{Gla93}. If fact, it is non-trivial that
  this relation is an equivalence, and that it satisfies the so-called
  \emph{stuttering property}. Once these properties have been
  established, it follows that the notion coincides with $\bbd$ of
  \cite{GW96}.

  In the remainder of this paper, we shall first, in
  Section~\ref{sec:bb}, briefly recapitulate the relational,
  coloured-trace, and modal characterisations of branching
  bisimilarity. Then, in Section~\ref{sec:relchar}, we shall discuss
  the condition proposed in \cite{Gla93} that can be added to the
  relational characterisation in order to make it divergence
  sensitive; we shall then also discuss several variants on this
  condition. In Section~\ref{sec:eqst} we establish that the
  relational characterisation of $\bbd$ all coincide, that they are
  equivalences and that they enjoy the stuttering property. In
  Section~\ref{sec:ctchar} we show that the relational
  characterisations of $\bbd$ coincide with the original definition of
  $\bbd$ in terms of coloured traces.  Finally, in
  Section~\ref{sec:modchar}, we shall establish agreement between the
  relational characterisation from \cite{Gla93}, the modal
  characterisation from \cite{Gla93}, and an alternative modal
  characterisation obtained by adding the divergence modality of
  \cite{Gla93} to the Hennessy-Milner logic with \emph{until} proposed
  in \cite{DV95}.

\section{Branching bisimilarity} \label{sec:bb}

  We presuppose a set
  $\Acts$ of \emph{actions} with a special element
  $\sact\in\Acts$, and we presuppose a \emph{labelled transition
    system} $(\States,\stepsym)$ with labels from $\Acts$, i.e.,
  $\States$ is a set of \emph{states} and
  $\stepsym\subseteq\States\times\Acts\times\States$ is a
  \emph{transition relation} on $\States$.
  Let $\states,\state{s'}\in\States$ and $\acta\in\Acts$.  We write
  $\states\step{\acta}\state{s'}$ for
  $(\states,\acta,\state{s'})\in\stepsym$ and we abbreviate the
  statement `$\states\step{\acta}\state{s'}$ or ($\acta=\sact$ and
  $\states=\state{s'}$)' by
  \plat{$\states\step{\opt{\acta}}\state{s'}$}.  We denote by
  $\plusstepssym$ the transitive closure of the binary relation
  $\step{\sact}$, and by $\sstepssym$ its reflexive-transitive
  closure.
  A \emph{path} from a state $\states$ is an alternating sequence
  $\states[0],a_1,\states[1],a_2,s_2,\dots,a_n,s_n$ of states and
  actions, such that $\states=\states[0]$ and
  $\states[k-1]\step{a_k}\states[k]$ for $k=1,\dots,n$.
  \linebreak
  A \emph{process} is given by a state $s$ in a labelled transition
  system, and encompasses all the states and transitions reachable
  from $s$.

  \paragraph{Relational characterisation}

  The definition of branching bisimilarity that is most widely used
  has a co-inductive flavour. 
  It defines when a binary relation on states preserves the
  behaviour of the associated processes. It then
  declares two states to be equivalent if there exists such a relation
  relating them. We shall refer to this kind of characterisation as a
  \emph{relational characterisation} of branching bisimilarity.

  \begin{definition} \label{def:bbisim} A symmetric binary relation
    $\brelsym$ on $\States$ is a \emph{branching bisimulation} if it
    satisfies the following condition for all
    $\states,\statet\in\States$ and $a\in\Acts$:
    \begin{enumerate}\itemsep 0pt
    \renewcommand{\theenumi}{T}
    \renewcommand{\labelenumi}{(\theenumi)}
    \item \label{cnd:stepsim}
      if $\states\brel\statet$ and $\states\step{\acta}\state{s'}$ for
      some state $\state{s'}$,
      then there exist states $\state{t'}$ and $\state{t''}$
      such that
        \plat{$\statet\ssteps{}\state{t''}\step{\opt{\acta}}\state{t'}$},
        $\states\brel\state{t''}$
      and
        $\state{s'}\brel\state{t'}$.
    \end{enumerate}
    We write $\states\bbisim\statet$ if there exists a branching
    bisimulation $\brelsym$ such that $\states\brel\statet$.
    The relation $\bbisim$ on states is referred to as (the relational
    characterisation of) branching bisimilarity.
  \end{definition}

  \noindent
  The relational characterisation of branching bisimilarity presented
  above is from \cite{Gla93}. As shown in \cite{Bas96,Gla93,GW96}, it
  yields the same concept of branching bisimilarity as the
  original definition in \cite{GW96}. The technical advantage of the
  above definition over the original definition is that the defined
  notion of branching bisimulation is \emph{compositional}: the
  composition of two branching bisimulations is again a branching
  bisimulation. Basten \cite{Bas96} gives an example showing that
  the condition used in the original definition of $\bbisimsym$ of
  \cite{GW96} fails to be compositional in this sense, and thus argued
  that establishing transitivity directly for the original definition
  is not straightforward.

  \paragraph{Coloured-trace characterisation}
  
  To substantiate their claim that branching bisimilarity indeed
  preserves the branching structure of processes, van Glabbeek
  and Weijland present in \cite{GW96} an alternative characterisation
  of the notion in terms of coloured traces. Below we repeat this
  characterisation.

  \begin{definition}\label{def:colourings}
  A \emph{colouring} is an equivalence on $S$.
  Given a colouring $\C$ and a state $s\in S$, the \emph{colour} $\C(s)$
  of $s$ is the equivalence class containing $s$.

  For $\pi = s_0,a_1,s_1,\dots,a_n,s_n$ a path from $s$, let $\C(\pi)$
  be the alternating sequence of colours and actions obtained from
  $\C(s_0),a_1,\C(s_1),\ldots,a_n,\C(s_n)$ by contracting all
  subsequences $C,\tau,C,\tau,\dots,\tau,C$ to $C$. The sequence
  $\C(\pi)$ is called a \emph{$\C$-coloured trace} of $s$.
  A colouring $\C$ is \emph{consistent} if two states of the
  same colour always have the same $\C$-coloured traces.

  We write $s\cc t$ if there exists a consistent colouring $\C$ with
  $\C(s)=\C(t)$.
  \end{definition}
 
  \noindent
  In \cite{GW96} it is proved that $\cc$ coincides with the relational
  characterisation $\bbisim$ of branching bisimilarity.

  \paragraph{Modal characterisation}

  A modal characterisation of a behavioural equivalence is a modal
  logic such that two processes are equivalent \IFF{} they satisfy the
  same formulas of the logic. The modal logic thus corresponding to a
  behavioural equivalence then allows one,
  for any two inequivalent processes, to formally express
  a behavioural property that distinguishes
  them. Whereas colourings or bisimulations are good tools to show
  that two processes are equivalent, modal formulas are better for
  proving inequivalence. The first modal characterisation
  of a behavioural equivalence is due to Hennessy and Milner
  \cite{HM85}. They provided a modal characterisation of
  (strong) bisimilarity on image-finite labelled transition systems,
  using a modal logic that is nowadays referred to as the
  \emph{Hennessy-Milner Logic}. The modal characterisations of
  branching bisimilarity presented below are adaptations
  of the Hennessy-Milner Logic.

  The class of formulas $\JBFrm$ of the modal logic for branching
  bisimilarity proposed in \cite{Gla93} is generated by the following
  grammar:
  \begin{equation} \label{eq:grammar}
    \frm\ ::=\
      \compl\frm\ \mid\
      \textstyle{\Meet\Frm}\ \mid\
      \frm\runtil[\acta]\frm
        \qquad\text{($\acta\in\Acts$,
               $\varphi\in\JBFrm$ and
               $\Frm\subseteq\JBFrm$).}
  \end{equation}
  In case the cardinality $|S|$ of the set of states of our labelled
  transition system is less than some infinite cardinal $\kappa$, we
  may require that $|\Frm|<\kappa$ in conjunctions, thus
  obtaining a \emph{set} of formulas rather than a proper class. We
  shall use the following standard abbreviations:
    $\true=\Meet\emptyset$,
    $\false=\compl\true$
  and
    $\bigvee\Frm=\compl\Meet\{\compl\frm\mid\frm\in\Frm\}$.

  We define when a formula $\frm$ is \emph{valid} in a state $\states$
  (notation: $\states\sat\frm$) inductively as follows:\vspace{-2pt}
  \begin{enumerate}
  \renewcommand{\theenumi}{\roman{enumi}}
  \renewcommand{\labelenumi}{(\theenumi)}
  \item
    $\states\sat\compl\frm$ \IFF{} $\states\not\sat\frm$;
  \item
    $\states\sat\Meet\Frm$
      \IFF{}
    $\states\sat\frm$ for all $\frm\in\Frm$;
  \item
    $\states\sat\frm\runtil[\acta]\frm[\psi]$
      \IFF{}
    there exist states $\state{s'}$ and $\state{s''}$ such that
      $\states\ssteps\state{s''}\step{\opt{\acta}}\state{s'}$,
      $\state{s''}\sat\frm$ and $\states'\sat\frm[\psi]$.
  \end{enumerate}
  Validity induces an equivalence on states: we define
  $\bbvalidsym\subseteq\States\times\States$ by
  \begin{equation*}
    \states\bbvalid\statet\quad
  \text{iff}\quad
    \all{\frm\in\JBFrm}.\
      {\states\sat\frm}\Leftrightarrow{\statet\sat\frm}
  \enskip.
  \end{equation*}
  In \cite{Gla93} it was shown that $\bbvalidsym$ coincides with
  $\bbisim$, that is, branching bisimilarity is characterised by the
  modal logic above.

  Clause (iii) in the definition of validity appears to be rather
  liberal. More stringent alternatives are obtained by using
  $\frm\funtil\frm[\psi]$ or $\frm\until\frm[\psi]$ instead of
  $\frm\runtil[\acta]\frm[\psi]$, with the following definitions:
  \begin{enumerate}
  \item[(iii$'$)] $\states\sat\frm\funtil\frm[\psi]$
          \IFF{}
        either $\acta=\sact$ and $\states\sat\frm[\psi]$, or
        there exists a sequence of states
        $\states[0],\dots,\states[n],\states[n+1]$ ($n\geq 0$)
        such that
          $\states=\states[0]
             \step{\sact}\cdots\step{\sact}\states[n]
             \step{\acta}\states[n+1]$,
          $\states[i]\sat\frm$ for all $i=0,\dots,n$
        and $\states[n+1]\sat\frm[\psi]$.
  \item[(iii$''$)] $\states\sat\frm\until[\acta]\frm[\psi]$
          \IFF{}
        there exists states
        $\states[0],\dots,\states[n],\states[n+1]$ ($n\geq 0$)
        such that
          $\states=\states[0]
             \step{\sact}\cdots\step{\sact}\states[n]
             \step{\opt{\acta}}\states[n+1]$,
          $\states[i]\sat\frm$ for all $i=0,\dots,n$
        and $\states[n+1]\sat\frm[\psi]$.
  \end{enumerate}
  The modality $\funtil$ stems from De Nicola \&
  Vaandrager \cite{DV95}. There it was shown, for labelled transition
  systems with \emph{bounded nondeterminism}, that branching
  bisimilarity, $\bbisimsym$, is characterised by the logic with
  negation, binary conjunction and this \emph{until}
  modality. The modality $\until$ is a common strengthening of
  $\funtil$ and the \emph{just-before} modality $\runtil[\acta]$ above; it
  was first considered in \cite{Gla93}.

  To be able to compare the expressiveness of modal logics, the
  following definitions are proposed by Laroussinie, Pinchinat \&
  Schnoebelen \cite{LPS95}.
  \begin{definition}\label{expressiveness}
    Two modal formulas $\varphi$ and $\psi$ that are interpreted on
    states of labelled transition systems are \emph{equivalent},
    written $\varphi \equi \psi$, if $s\sat\varphi
    \Leftrightarrow s\sat\psi$ for all states $s$ in all labelled
    transition systems. Two modal logics are equally expressive if for
    every formula in the one there is an equivalent formula in
    the other.
  \end{definition}

\noindent
  As remarked in \cite{Gla93}, the modalities $\funtil$
  and $\until$ are equally expressive, since
  \begin{alignat*}{3}
    & \varphi \funtil[\tau] \psi
    &&\ ~\equi~\
      \psi \vee \varphi \until[\tau] \psi
  \enskip&&\ ,\\
    & \varphi \until[\tau] \psi
    &&\ ~\equi~\
      \varphi \wedge \varphi \funtil[\tau] \psi
  \enskip&&\ \text{and}\\
    & \varphi \until[\acta] \psi
    &&\ ~\equi~\
      \varphi \funtil \psi
  &&\ \text{for all $a\neq\silent$}.
  \end{alignat*}
  Note that the modality $\runtil$ can be expressed in terms of
  $\until$:
  \begin{alignat*}{3}
    & \varphi \runtil \psi
    &&\ ~\equi~\
      \top\until[\tau](\varphi \until \psi)\enskip.
  \end{alignat*}

  \noindent
  Laroussinie, Pinchinat \& Schnoebelen established in \cite{LPS95}
  that the modal logic with negation, binary conjunction and $\runtil$
  from \cite{Gla93} and the logic with negation, binary conjunction
  and $\funtil$ from \cite{DV95} are equally expressive.

  \section[Adding divergence]
          {Relational characterisations of \bbd{}}
          \label{sec:relchar}

  The notion branching bisimilarity discussed in the previous section
  abstracts from divergence (i.e, infinite internal activity). In the
  remainder of this paper, we discuss a refinement of the notion of
  branching bisimulation equivalence that takes
  divergence into account. In this section we present several
  conditions that can be added to the notion of branching bisimulation
  in order to make it divergence sensitive. The induced notions of
  branching bisimilarity with explicit divergence will all turn out to
  be equivalent.

  \begin{definition}\cite{Gla93} \label{def:bbisimd}
    A symmetric binary relation $\brelsym$ on $\States$ is a
    \emph{branching bisimulation with explicit divergence} if it is a
    branching bisimulation (i.e., it satisfies condition
    (\ref{cnd:stepsim}) of Definition~\ref{def:bbisim}) and in addition
    satisfies the following condition for all
    $\states,\statet\in\States$ and $a\in\Acts$:
    \begin{enumerate}\itemsep 0pt
  \renewcommand{\theenumi}{D}
  \renewcommand{\labelenumi}{(\theenumi)}
  \item \label{cnd:rvgdivsim}
    if $\states\brel\statet$
    and there is an infinite sequence of states
      $(\states[k])_{k\in\N}$
    such that
      $\states=\states[0]$,
      $\states[k]\step{\sact}\states[k+1]$
    and
      $\states[k]\brel\statet$ for all $k\in\N$,
    then there exists an infinite sequence of states
      $(\statet[\ell])_{\ell\in\N}$
    such that
      $\statet=\statet[0]$,
      $\statet[\ell]\step{\sact}\statet[\ell+1]$ for all $\ell\in\N$,
    and
      $\states[k]\brel\statet[\ell]$ for all $k,\ell\in\N$.
    \end{enumerate}
    We write $\states\rbbisimd\statet$ if there exists a branching
    bisimulation with explicit divergence $\brelsym$ such that
    $\states\brel\statet$.
  \end{definition}

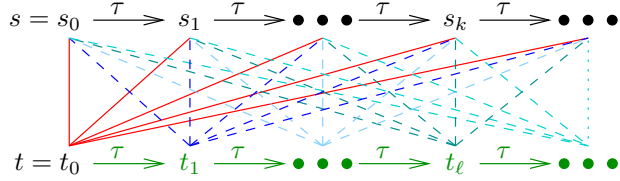
\begin{figure}[htb]
\begin{center}\input{rvgdivsim.pstex_t}\end{center}
\caption{Condition~(\ref{cnd:rvgdivsim}).}
\label{fig:rvgdivsim}
\end{figure}

\noindent
Figure~\ref{fig:rvgdivsim} illustrates condition
(\ref{cnd:rvgdivsim}).  In \cite{Gla93} it was claimed that the notion
$\rbbisimd$ defined above coincides with \emph{branching bisimilarity
  with explicit divergence} as defined earlier in
\cite{GW96}. In this paper we will substantiate this claim.  On the
way to this end, we need to show that $\rbbisimd$ is an equivalence
and has the so-called \emph{stuttering property}.

The difficulty in proving that $\rbbisimd$ is an equivalence is in
establishing transitivity. Basten's proof in \cite{Bas96} that
$\bbisimsym$ (i.e., branching bisimilarity without explicit
divergence) is transitive, is obtained as an immediate consequence of
the fact that whenever two binary relations $\brel_1$ and $\brel_2$
satisfy (\ref{cnd:stepsim}), then so does their composition
$\brel_1\relcomp\brel_2$ (see Lemma~\ref{lem:relcomp} below).
The condition (\ref{cnd:rvgdivsim}) fails to be compositional, as we
show in the following example.

\begin{figure}[htb]
\begin{center}\input{incompositional_original.pstex_t}\end{center}
\caption{The composition of branching bisimulations with explicit
  divergence is not a branching bisimulation with explicit divergence.}
\label{fig:incompositional_original}
\end{figure}
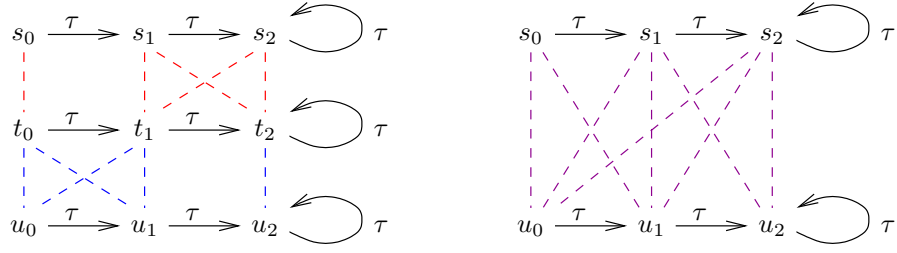

  \begin{example} \label{exa:incompositionality}
    Consider the labelled transition system depicted on the
    left-hand side of Figure~\ref{fig:incompositional_original} together with
    the branching bisimulations with explicit divergence
    \begin{gather*}
      \brel_1=\{(\states[0],\statet[0]),(\statet[0],\states[0]),
                (\states[1],\statet[1]),(\statet[1],\states[1]),
                (\states[2],\statet[2]),(\statet[2],\states[2]),
                (\states[1],\statet[2]),(\statet[2],\states[1]),
                (\states[2],\statet[1]),(\statet[1],\states[2])
              \}~~~~~\hspace{-1.67pt}\text{and}\\
      \brel_2=\{(\statet[0],\stateu[0]),(\stateu[0],\statet[0]),
                (\statet[1],\stateu[1]),(\stateu[1],\statet[1]),
                (\statet[2],\stateu[2]),(\stateu[2],\statet[2]),
                (\statet[0],\stateu[1]),(\stateu[1],\statet[0]),
                (\statet[1],\stateu[0]),(\stateu[0],\statet[1])
              \}
    \enskip.
    \end{gather*}
    The composition $\brelsym = {\brel_1\relcomp\brel_2}$ on the
    relevant fragment is depicted on the right-hand side of
    Figure~\ref{fig:incompositional_original}. 
    Note that
      $\states[0]$
    gives rise to a divergence
    {\ofwhicheverystateis} related by $\brel$ to $\stateu[0]$.
    However, since $\states[0]$ and $\stateu[2]$ are not related
    according to $\brel$, there is no divergence from $\stateu[0]$
    {\ofwhicheverystateis} related to every state on the divergence from
    $\states[0]$. We conclude that $\brel$ does not satisfy the
    condition (\ref{cnd:rvgdivsim}).    
  \end{example}

  \noindent
  Our proof that $\rbbisimd$ is an equivalence proceeds along the same
  lines as Basten's proof in \cite{Bas96} that $\bbisim$ is an
  equivalence: we replace (\ref{cnd:rvgdivsim}) by an alternative
  divergence condition that is compositional,
  prove that the resulting notion of bisimilarity is an equivalence,
  and then establish that it coincides with $\rbbisimd$.
  In the remainder of this section, we shall arrive at our
  compositional alternative for (\ref{cnd:rvgdivsim}) through a series
  of adaptations of (\ref{cnd:rvgdivsim}).

  First, we observe that (\ref{cnd:rvgdivsim}) has a technically
  convenient reformulation: instead of requiring the existence of a
  \emph{divergence} from $\statet$ all the states {of which} enjoy
  certain properties, it suffices to require that there exists a
  \emph{state} reachable from $\statet$ by a single $\sact$-transition
  with these properties. Formally, the reformulation of
  (\ref{cnd:rvgdivsim}) is:
  \begin{enumerate}\itemsep 0pt
  \renewcommand{\theenumi}{D$_0\hspace{-1pt}$}
  \renewcommand{\labelenumi}{(\theenumi)}
  \item \label{cnd:rvgdivsimshort}
    if $\states\brel\statet$
    and there is an infinite sequence of states
      $(\states[k])_{k\in\N}$
    such that
      $\states=\states[0]$,
      $\states[k]\step{\sact}\states[k+1]$
    and
      $\states[k]\brel\statet$ for all $k\in\N$,
    then there exists a state $\state{t'}$
    such that
      $\statet\step{\sact}\state{t'}$
    and
      $\states[k]\brel\state{t'}$ for all $k\in\N$.
  \end{enumerate}

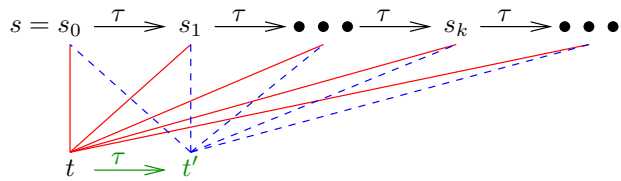
\begin{figure}[htb]
\begin{center}\input{rvgdivsimshort.pstex_t}\end{center}
\caption{Condition~(\ref{cnd:rvgdivsimshort}).}
\label{fig:rvgdivsimshort}
\end{figure}

  \noindent
  Figure~\ref{fig:rvgdivsimshort} illustrates condition
  (\ref{cnd:rvgdivsimshort}).
  If a binary relation satisfies (\ref{cnd:rvgdivsimshort}), then the
  divergence from $\statet$ required by (\ref{cnd:rvgdivsim}) can be
  inductively constructed. (We omit the inductive construction here;
  the proof of Proposition~\ref{prop:short} below contains a very
  similar inductive construction.)

  For our next adaptation we observe that (\ref{cnd:rvgdivsimshort})
  has some redundancy.  Note that it requires $\state{t'}$ to be
  related to \emph{every} state on the divergence from
  $\states$. However, the universal quantification in the conclusion
  can be relaxed to an existential quantification: it suffices to
  require that $\statet$ has an immediate $\sact$-successor that is
  related to \emph{some} state on the divergence from $\states$. The
  requirement can be expressed as follows:
  \begin{enumerate}\itemsep 0pt
   \renewcommand{\theenumi}{D$_1\hspace{-1pt}$}
   \renewcommand{\labelenumi}{(\theenumi)}
   \addtocounter{enumi}{1}
    \item \label{cnd:onestepdivsim}
      if $\states\brel\statet$
      and there is an infinite sequence of states
        $(\states[k])_{k\in\N}$
      such that
         $\states=\states[0]$,
         $\states[k]\step{\sact}\states[k+1]$
      and
         $\states[k]\brel\statet$ for all $k\in\N$,
      then there exists a state $\state{t'}$
      such that
        \plat{$\statet\step{\sact}\state{t'}$}
      and
        $\states[k]\brel\state{t'}$ for some $k\in\N$.
  \end{enumerate}

\begin{figure}[htb]
\begin{center}\input{onestepdivsim.pstex_t}\end{center}
\caption{Condition~(\ref{cnd:onestepdivsim}).}
\label{fig:divcond0}
\end{figure}
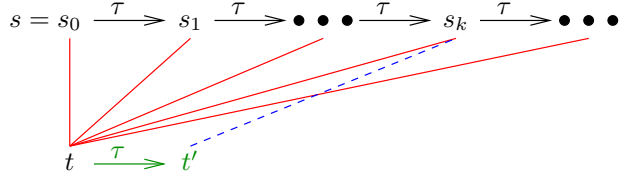

  \noindent
  Condition (\ref{cnd:onestepdivsim}) appears in the definition of
  divergence-sensitive stuttering simulation of Nejati \cite{Nej03}.
  It is illustrated in Figure~\ref{fig:divcond0}. We write
  $\states\onestepbbisimd\statet$ if there exists a symmetric binary
  relation $\brelsym$ satisfying (\ref{cnd:stepsim}) and
  (\ref{cnd:onestepdivsim}) such that $\states\brel\statet$. Note that
  every relation satisfying (\ref{cnd:rvgdivsim}) also satisfies
  (\ref{cnd:onestepdivsim}), so it follows that $\rbbisimdsym
  \subseteq \onestepbbisimdsym$.

  The following example illustrates that condition
  (\ref{cnd:onestepdivsim}) is still not compositional, not even if
  the composed relations satisfy (\ref{cnd:stepsim}).

  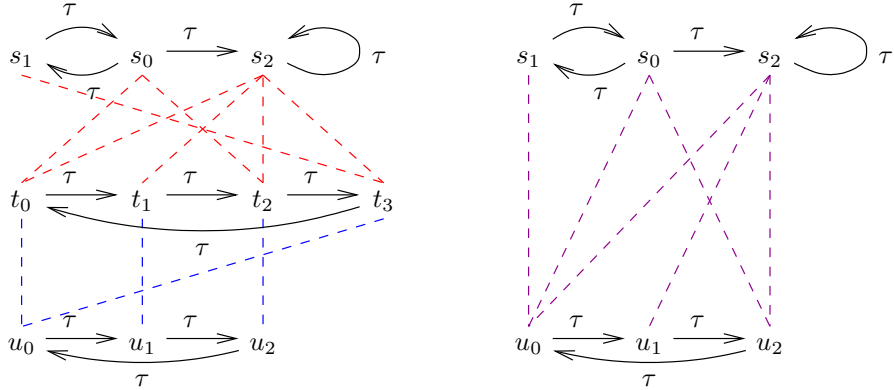
\begin{figure}[htb]
  \begin{center}\input{incompositional_onestepdivsim.pstex_t}\end{center}
  \caption{The composition of binary relations satisfying
    (\ref{cnd:stepsim}) and (\ref{cnd:onestepdivsim}) does not necessarily satisfy (\ref{cnd:onestepdivsim}). }
  \label{fig:incompositional_onestepdivsim}
  \end{figure}

  \begin{example}
    Consider the labelled transition system depicted on the
    left-hand side of Figure~\ref{fig:incompositional_onestepdivsim}
    together with two binary relations satisfying
    (\ref{cnd:stepsim}) and (\ref{cnd:onestepdivsim}):
    \begin{gather*}
      \brel_1=\{(\states[0],\statet[0]),(\statet[0],\states[0]),
                (\states[0],\statet[2]),(\statet[2],\states[0]),
                (\states[1],\statet[3]),(\statet[3],\states[1])
              \}\cup
              \{(\states[2],\statet[i]),(\statet[i],\states[2])
                   \mid 0\leq i \leq 3
              \}\ \text{and}\\
      \brel_2=\{(\statet[i],\stateu[i]),(\stateu[i],\statet[i])
                  \mid 0\leq i \leq 2
              \}\cup
              \{(\statet[3],\stateu[0]),(\stateu[0],\statet[3])
              \}
    \enskip.
    \end{gather*}
    Note that, since $\states[1]$ is not $\brelsym_1$-related to
    $\statet[0]$, the divergence
      $\states[0]
         \step{\sact}\states[1]
         \step{\sact}\states[0]
         \step{\sact}\states[1]
         \step{\sact}\cdots$
    need not be simulated by $\statet[0]$ in such a way that
    $\statet[1]$ is related to either $\states[0]$ or $\states[1]$.

    Now consider the composition
       $\brelsym = {\brel_1\relcomp\brel_2}$.
    Both $\states[0]$ and $\states[1]$ are $\brelsym$-related to
    $\stateu[0]$, whereas the state $\stateu[1]$ is not
    $\brelsym$-related to $\states[0]$ nor to $\states[1]$.
    We conclude that $\brelsym$ does not satisfy
    (\ref{cnd:onestepdivsim}).
  \end{example}

\noindent
  The culprit in the preceding example appears to be the fact that
  (\ref{cnd:onestepdivsim}) only considers divergences from $\states$
  {\ofwhicheverystateis} related to $\statet$. Our second alternative
  omits this restriction. It considers \emph{every} divergence from
  $\states$ and requires that it is simulated by $\statet$. 
  \begin{enumerate}\itemsep 0pt
   \renewcommand{\labelenumi}{(\theenumi)}
   \renewcommand{\theenumi}{D$_2\hspace{-1pt}$}
   \item \label{cnd:newdivsimshort}
       if $\states\brel\statet$
       and there is an infinite sequence of states
         $(\states[k])_{k\in\N}$
       such that
          $\states=\states[0]$
      and
          $\states[k]\step{\sact}\states[k+1]$ for all $k\in\N$,
      then there exists a state $\state{t'}$
      such that
        $\statet\step{\sact}\state{t'}$
      and
        $\states[k]\brel\state{t'}$ for some $k\in\N$.
  \end{enumerate}

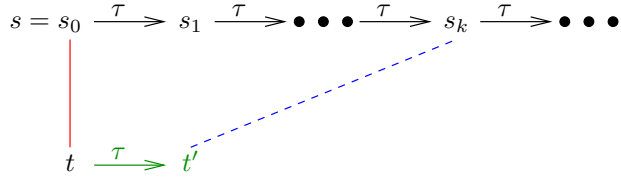
\begin{figure}[htb]
\begin{center}\input{newdivsimshort.pstex_t}\end{center}
\caption{Condition~(\ref{cnd:newdivsimshort}).}
\label{fig:newdivsimshort}
\end{figure}

\noindent
  Figure~\ref{fig:newdivsimshort} illustrates condition
  (\ref{cnd:newdivsimshort}).
  In contrast to the preceding divergence conditions, it does have the
  property that if two relations both satisfy it, then so does their
  relational composition. However, to facilitate a direct proof of this
  property, it is technically convenient to reformulate condition
  (\ref{cnd:newdivsimshort}) such that it requires a divergence from
  $\statet$ rather than just one $\sact$-step:

  \begin{enumerate}\itemsep 0pt
   \renewcommand{\labelenumi}{(\theenumi)}
   \renewcommand{\theenumi}{D$_3\hspace{-1pt}$}
   \item \label{cnd:newdivsim}
       if $\states\brel\statet$
       and there is an infinite sequence of states
         $(\states[k])_{k\in\N}$
       such that
          $\states=\states[0]$
       and
          $\states[k]\step{\sact}\states[k+1]$ for all $k\in\N$,
       then there exist an infinite sequence of states
         $(\statet[\ell])_{\ell\in\N}$
       and a mapping
         $\statemap: {\N\rightarrow\N}$
       such that
         $\statet=\statet[0]$,
         $\statet[\ell]\step{\sact}\statet[\ell+1]$
       and
         $\states[\statemap(\ell)]\brel\statet[\ell]$ for all $\ell\in\N$.
  \end{enumerate}

\begin{figure}[htb]
\begin{center}\input{newdivsim.pstex_t}\end{center}
\caption{Condition~(\ref{cnd:newdivsim}).}
\label{fig:newdivsim}
\end{figure}
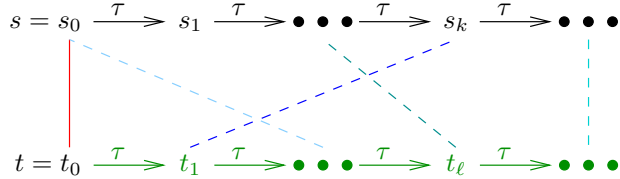

\noindent
  Figure~\ref{fig:newdivsim} illustrates condition (\ref{cnd:newdivsim}).

  \begin{proposition} \label{prop:short}
    A binary relation $\brelsym$ satisfies
      {\rm (\ref{cnd:newdivsimshort})}
    iff it satisfies
      {\rm (\ref{cnd:newdivsim})}.
  \end{proposition}
  \begin{proof}
    The implication from right to left is trivial.
    For the implication from left to right, suppose that $\brelsym$
    satisfies (\ref{cnd:newdivsimshort}) and that
    $\states\brel\statet$, and consider an infinite sequence of states
    $(\states[k])_{k\in\N}$ such that
      $\states=\states[0]$
    and
      $\states[k]\step{\sact}\states[k+1]$ for all $k\in\N$. We construct
    an infinite sequence of states
      $(\statet[\ell])_{\ell\in\N}$
    and a mapping
      $\statemap: {\N \rightarrow \N}$
    such that
      $\statet=\statet[0]$,
      $\statet[\ell]\step{\sact}\statet[\ell+1]$
    and
      $\states[\statemap(\ell)]\brel\statet[\ell]$ for all $\ell\in\N$.

    The infinite sequence $(\statet[\ell])_{\ell\in\N}$ and the mapping
    $\statemap:\N\rightarrow\N$ can be defined simultaneously by
    induction on $l$:
    \begin{enumerate}\itemsep 0pt
      \item We define $\statet[0]=\statet$ and $\statemap(0)=0$;
        it then clearly holds that
          $\states[\statemap(0)]\brel\statet[0]$.
      \item Suppose that the sequence $(\statet[\ell])_{\ell\in\N}$
        and the mapping $\statemap:\N\rightarrow\N$ have been defined
        up to $\ell$. Then, in particular,
        $\states[\statemap(\ell)]\brel\statet[\ell]$. 
        Since $(\states[\statemap(\ell)+k])_{k\in\N}$ is an infinite
        sequence such that
          $\states[\statemap(\ell)+k]\step{\sact}\states[\statemap(\ell)+k+1]$
        for all $k\in\N$, by (\ref{cnd:newdivsimshort}) there exists
        $\state{t'}$ such that $\statet[\ell]\step{\sact}\state{t'}$ and
        $\states[\statemap(\ell)+k']\brel\state{t'}$ for some
        $k'\in\N$.
        We define $\statet[\ell+1]=\state{t'}$ and
        $\statemap(\ell+1)=k'$.
  \qed
    \end{enumerate}
  \end{proof}

  \noindent
  We write $\states\wbbisimd\statet$ if 
  there exists a symmetric binary relation $\brelsym$ satisfying
  (\ref{cnd:stepsim}) and (\ref{cnd:newdivsim}) such that
  \mbox{$\states\brel\statet$}. Note that (\ref{cnd:onestepdivsim}) is
  a weaker requirement than (\ref{cnd:newdivsimshort}), and hence, by
  Proposition~\ref{prop:short}, than (\ref{cnd:newdivsim}).
  It follows that $\wbbisimdsym \subseteq \onestepbbisimdsym$. Also
  note that (\ref{cnd:newdivsimshort}) and (\ref{cnd:newdivsim}) on
  the one hand and (\ref{cnd:rvgdivsim}) and (\ref{cnd:rvgdivsimshort})
  on the other hand are incomparable.

  Using that (\ref{cnd:newdivsim}) is compositional, it will be
  straightforward to establish that $\wbbisimdsym$ is an equivalence.
  Then, it remains to establish that $\rbbisimd$ and $\wbbisimd$
  coincide. We shall prove that $\wbbisimd$ is included in $\rbbisimd$
  by establishing that $\wbbisimd$ is a branching bisimulation with
  explicit divergence; that $\wbbisimd$ is an equivalence is crucial
  in the proof of this property.  Instead of proving the converse
  inclusion directly, we obtain a stronger result by establishing that
  a notion of bisimilarity defined using a weaker divergence condition
  and therefore including $\rbbisimd$, is included in $\wbbisimdsym$.
  The weakest divergence condition we encountered so far is
  (\ref{cnd:onestepdivsim}). It is, however, possible to further
  weaken (\ref{cnd:onestepdivsim}): instead of requiring that
  $\state{t'}$ is an immediate $\sact$-successor, it is enough require
  that $\state{t'}$ can be reached from $\statet$ by one or more
  $\sact$-transitions. Formally,

  \begin{enumerate}\itemsep 0pt
   \renewcommand{\theenumi}{D$_4\hspace{-1pt}$}
   \renewcommand{\labelenumi}{(\theenumi)}
   \addtocounter{enumi}{1}
    \item \label{cnd:divsim}
      if $\states\brel\statet$
      and there is an infinite sequence of states
        $(\states[k])_{k\in\N}$
      such that
         $\states=\states[0]$,
         $\states[k]\step{\sact}\states[k+1]$
      and
         $\states[k]\brel\statet$ for all $k\in\N$,
      then there exists a state $\state{t'}$
      such that
        \plat{$\statet\plussteps\state{t'}$}
      and
        $\states[k]\brel\state{t'}$ for some $k\in\N$.
  \end{enumerate}

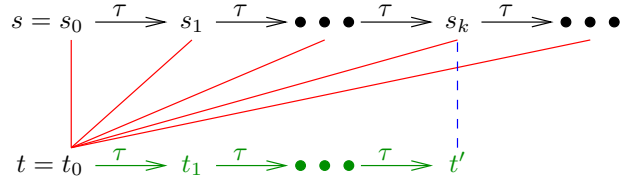
\begin{figure}[htb]
\begin{center}\input{divsim.pstex_t}\end{center}
\caption{Condition~(\ref{cnd:divsim}).}
\label{fig:divsim}
\end{figure}

\noindent
  Figure~\ref{fig:divsim} illustrates condition (\ref{cnd:divsim}).
  We write $\states\bbisimd\statet$ if there exists a symmetric binary
  relation $\brelsym$ satisfying (\ref{cnd:stepsim}) and
  (\ref{cnd:divsim}) such that 
  $\states\brel\statet$. Clearly,
    $\onestepbbisimdsym\subseteq\bbisimdsym$,
  and hence also $\wbbisimdsym\subseteq\bbisimdsym$ and
  $\rbbisimdsym\subseteq\bbisimdsym$.

  In the next section we shall also prove that
    $\bbisimdsym\subseteq\wbbisimdsym$.
  A crucial tool in our proof of this inclusion will be the notion of
  \emph{stuttering closure} of a binary relation $\brel$ on
  states. The stuttering closure of $\brel$ enjoys the so-called
  \emph{stuttering property}: if from state $\states$ a state
  $\state{s'}$ can be reached through a sequence of
  $\sact$-transitions, and both $\states$ and $\state{s'}$ are 
  $\brel$-related to the same state $\statet$, then all intermediate
  states between $\states$ and $\state{s'}$ are $\brel$-related to
  $\statet$ too.
  We shall prove a lemma to the effect that if a binary relation on
  states satisfies (\ref{cnd:stepsim}) and (\ref{cnd:divsim}), then
  its stuttering closure satisfies (\ref{cnd:stepsim}) and
  (\ref{cnd:newdivsim}), and use it to establish the inclusion
  $\bbisimdsym\subseteq\wbbisimdsym$. An easy corollary of the lemma
  is that $\bbisimdsym$ has the stuttering property.
  Here our proof also has a similarity with Basten's proof in
  \cite{Bas96}; in his proof that the notions of branching
  bisimilarity induced by (\ref{cnd:stepsim}) and by the original
  condition used in \cite{GW96} coincide, establishing the stuttering
  property is a crucial step.

\begin{figure}[htb]
\begin{center}\input{inclusiongraph.pstex_t}\end{center}
\caption{Inclusion graph.}
\label{fig:inclusion}
\end{figure}
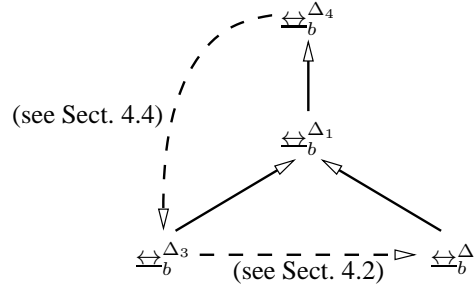

  Figure~\ref{fig:inclusion} shows some inclusions between the
  different versions of branching bisimilarity with explicit
  divergence. (Note that we never defined $\rbbisimdzero$ and
  $\wbbisimdtwo$, as these would be the same as $\rbbisimd$ and
  $\wbbisimd$, respectively.)
  The solid arrows indicate inclusions that have already
  been argued for above; the dashed arrows indicate inclusions that
  will be established below.

  \begin{remark}
    We shall establish in the next section that
    $\rbbisimdsym=\bbisimdsym$. Note that, once we have this, we can
    replace the second condition of Definition~\ref{def:bbisimd} by
    any \emph{interpolant} of (\ref{cnd:rvgdivsim}) and
    (\ref{cnd:divsim}), i.e., any condition that is implied by
    (\ref{cnd:rvgdivsim}) and implies (\ref{cnd:divsim}), and end up
    with the same equivalence.  For instance, we could replace it by
    condition (\ref{cnd:onestepdivsim}), or by the condition of Gerth,
    Kuiper, Peled \& Penczek~\cite{GKPP99}:
    \begin{enumerate}\itemsep 0pt \item[]
      if $\states\brel\statet$
      and there is an infinite sequence of states
        $(\state{s_k})_{k\in\N}$
      such that
         $\states=\state{s_0}$,
         $\state{s_k}\step{\sact}\state{s_{k+1}}$
      and
         $\state{s_k}\brel\statet$ for all $k\in\N$,
      then there exists a state $\state{t'}$ such that
        $\statet\step{\sact}\state{t'}$ and
        $\state{s_k}\brel\state{t'}$ for some $k>0$.
    \end{enumerate}
    Similarly, we will prove that $\wbbisimdsym=\bbisimdsym$, and so
    we can replace the second condition of
    Definition~\ref{def:bbisimd} by an interpolant of
    (\ref{cnd:newdivsim}) and (\ref{cnd:divsim}). For instance, the
    condition
    \begin{enumerate}\itemsep 0pt
    \item[]
      if $\states\brel\statet$
      and there is an infinite sequence of states
        $(\states[k])_{k\in\N}$
      such that
         $\states=\states[0]$
      and
         $\states[k]\step{\sact}\states[k+1]$ for all $k\in\N$,
      then there exists a state $\state{t'}$
      such that
        \plat{$\statet\plussteps\state{t'}$}
      and
        $\states[k]\brel\state{t'}$ for some $k\geq 0$
    \end{enumerate}
    is a convenient interpolant of (\ref{cnd:newdivsim}) and
    (\ref{cnd:divsim}) to use when showing that two states are
    branching bisimulation equivalent with explicit divergence.
  \end{remark}

  \section[Equivalence and stuttering]
  {\bbd{} is an equivalence with the stuttering property} \label{sec:eqst}

  Our goal is now to establish that the relational characterisations
  of branching bisimilarity with explicit divergence introduced in the
  previous section all coincide, that they are equivalences and that
  they enjoy the stuttering property. To this end, we first show that
  $\wbbisimdsym$ is an equivalence relation; condition
  (\ref{cnd:newdivsim}) will enable a direct proof of this fact. Using
  that $\wbbisimdsym$ is an equivalence, we obtain $\wbbisimdsym
  \subseteq \rbbisimdsym$.  Then, we define the notion of
  \emph{stuttering closure} and use it to establish $\bbisimdsym
  \subseteq \wbbisimdsym$. Together with the observation $\rbbisimdsym
  \subseteq \bbisimdsym$ made above, the cycle of inclusions yields
  that the relations $\rbbisimdsym$, $\wbbisimdsym$ and $\bbisimdsym$
  coincide. It then follows that $\rbbisimdsym$ is an equivalence.  We
  have not been able to find a less roundabout way to obtain this
  result. The intermediate results needed for the equivalence proof
  also yields that $\rbbisimdsym$ has the stuttering property.

  \subsection[]{$\wbbisimdsym$ is an equivalence} \label{subsec:wbbisimdeq}

The proofs below are rather straightforward. Nevertheless, the
proof strategy employed for Lemmas \ref{lem:relunion}
and~\ref{lem:relcomp} would fail for $\rbbisimdsym$, $\onestepbbisimd$
and $\bbisimdsym$. It is for this reason that we present all detail.

  \begin{lemma} \label{lem:relunion}
    Let $\{\brelsym_i\mid i\in I\}$ be a family of binary relations.
  \begin{enumerate}\itemsep 0pt
  \renewcommand{\theenumi}{\roman{enumi}}
  \renewcommand{\labelenumi}{(\theenumi)}
  \item
    If $\brelsym_i$ satisfies {\rm (\ref{cnd:stepsim})} for all $i\in
    I$, then so does the union $\bigcup_{i\in I}\brelsym_i$.
  \item \label{lem:relunion:item:newdivsim}
    If $\brelsym_i$ satisfies {\rm (\ref{cnd:newdivsim})} for all
    $i\in I$, then so does the union $\bigcup_{i\in I}\brelsym_i$.
  \end{enumerate}
  \end{lemma}
  \begin{proof}
  Let $\brelsym=\bigcup_{i\in I}\brelsym_i$.
  \begin{enumerate}\itemsep 0pt
  \renewcommand{\theenumi}{\roman{enumi}}
  \renewcommand{\labelenumi}{(\theenumi)}
  \item Suppose that $\brelsym_i$ satisfies
    (\ref{cnd:stepsim}) for all $i \in I$.
    To prove that $\brelsym$ also satisfies (\ref{cnd:stepsim}),
    suppose that
      $\states\brel\statet$
    and
      $\states\step{\acta}\state{s'}$ for some state $\state{s'}$.
    Then $\states\brel_i\statet$ for some $i\in I$.
    Since $\brelsym_i$ satisfies
    (\ref{cnd:stepsim}), it follows that there are states
    $\state{t'}$ and $\state{t''}$ such that
      \plat{$\statet\ssteps{}\state{t''}\step{\opt{\acta}}\state{t'}$},
      $\states\brel_i\state{t''}$
    and
      $\state{s'}\brel_i\state{t'}$,
    and hence $\states\brel\state{t''}$ and
    $\state{s'}\brel\state{t'}$.
   \item Suppose that $\brelsym_i$ satisfies
     (\ref{cnd:newdivsim}) for all $i \in I$.
     To prove that $\brelsym$ satisfies (\ref{cnd:newdivsim}), suppose
     that $\states\brel\statet$ and that there is an infinite sequence
     of states $(\states[k])_{k\in\N}$ such that
       $\states=\states[0]$
     and
       $\states[k]\step{\sact}\states[k+1]$.
     From $\states\brel\statet$ it follows that
     $\states\brel_i\statet$ for some $i \in I$.
     By (\ref{cnd:newdivsim}) there exist an infinite sequence of states
       $(\statet[\ell])_{\ell\in\N}$
     and a mapping
       $\statemap:\N\rightarrow\N$
     such that
       $\statet=\statet[0]$,
       $\statet[\ell]\step{\sact}\statet[\ell+1]$
     and
       $\states[\statemap(\ell)]\brel_i\statet[\ell]$ for all $\ell\in\N$,
     and from the latter it follows that
       $\states[\statemap(\ell)]\brel\statet[\ell]$
     for all $\ell\in\N$.
  \qed
  \end{enumerate}
  \end{proof}

  \begin{lemma} \label{lem:isync}
    Let $\brelsym$ be a binary relation that satisfies
    {\rm (\ref{cnd:stepsim})}.
    If $\states\brel\statet$ and $\states\ssteps{}\state{s'}$,
    then there is a state $\state{t'}$ such that
      $\statet\ssteps{}\state{t'}$
    and
      $\state{s'}\brel\state{t'}$.
  \end{lemma}
  \begin{proof}
    Let $\states[0],\dots,\states[n]$ be states such that
      $\states=
         \states[0]\step{\sact}\cdots\step{\sact}\states[n]
           =\state{s'}$.
    By (\ref{cnd:stepsim}) and a straightforward
    induction on $n$ there exist states
      $\statet[0],\dots,\statet[n]$
    such that
      $\statet=
         \statet[0]\ssteps{}\cdots\ssteps{}\statet[n]
           =\state{t'}$
    and $\states[i]\brel\statet[i]$ for all $i\leq n$.
  \end{proof}

  \begin{lemma} \label{lem:relcomp}
    Let $\brelsym_1$ and $\brelsym_2$ be binary relations.
  \begin{enumerate}\itemsep 0pt
  \renewcommand{\theenumi}{\roman{enumi}}
  \renewcommand{\labelenumi}{(\theenumi)}
  \item
    If $\brelsym_1$ and $\brelsym_2$ both satisfy
    {\rm (\ref{cnd:stepsim})}, then so does their composition
    $\brelsym_1\relcomp\brelsym_2$.
  \item
    If $\brelsym_1$ and $\brelsym_2$ both satisfy
    {\rm (\ref{cnd:newdivsim})}, then so does their composition
    $\brelsym_1\relcomp\brelsym_2$.
  \end{enumerate}
  \end{lemma}
  \begin{proof}
    Let $\brelsym=\brelsym_1\relcomp\brelsym_2$.
  \begin{enumerate}\itemsep 0pt
  \renewcommand{\theenumi}{\roman{enumi}}
  \renewcommand{\labelenumi}{(\theenumi)}
  \item To prove that $\brel$ satisfies
    (\ref{cnd:stepsim}), suppose $\states\brel\stateu$ and
    $\states\step{\acta}\state{s'}$.
    Then there exists a state $\statet$ such that
        $\states\brel_1\statet$ and $\statet\brel_2\stateu$.
    Since $\brelsym_1$ satisfies (\ref{cnd:stepsim}), there exist
    states $\state{t'}$ and $\state{t''}$ such that
        \plat{$\statet\ssteps{}\state{t''}\step{\opt{\acta}}\state{t'}$},
        $\states\brel_1\state{t''}$
    and
        $\state{s'}\brel_1\state{t'}$.
    By Lemma \ref{lem:isync} there is a state $\state{u''}$
    such that
      $\stateu\ssteps{}\state{u''}$
    and
      $\state{t''}\brel_2\state{u''}$.
    We now distinguish two cases:
    \begin{enumerate}\itemsep 0pt
    \item Suppose that $\acta=\sact$ and $\state{t''}=\state{t'}$.
      Then $\stateu\ssteps{}\state{u''}\step{\opt{\acta}}\state{u''}$,
      from
        $\states\brel_1\state{t''}$ and $\state{t''}\brel_2\state{u''}$
      it follows that
        $\states\brel\state{u''}$,
      and from
        $\state{s'}\brel_1\state{t'}$ and $\state{t'}\brel_2\state{u''}$
      it follows that
        $\state{s'}\brel\state{u''}$.
    \item Suppose that $\statet''\step{\acta}\state{t'}$.
      Then there exist states $\state{u'''}$ and $\state{u'}$ such
      that $\state{u''}\ssteps{}\state{u'''}\step{\opt{\acta}}\state{u'}$,
        $\state{t''}\brel_2\state{u'''}$
      and
        $\state{t'}\brel_2\state{u'}$.
      So, \plat{$\state{u}\ssteps{}\state{u'''}\step{\opt{\acta}}\state{u'}$},
      from
        $\states\brel_1\state{t''}$ and $\state{t''}\brel_2\state{u'''}$
      it follows that
        $\states \brel \state{u'''}$,
      and from
        $\state{s'}\brel_1\state{t'}$ and $\state{t'}\brel_2\state{u'}$
      it follows that
        $\state{s'}\brel\state{u'}$.
    \end{enumerate}
  \item To prove that $\brel$ satisfies (\ref{cnd:newdivsim}),
    suppose that
        $\states\brel\stateu$
    and that there is an infinite sequence of states
        $(\states[k])_{k\in\N}$
    such that
        $\states=\states[0]$,
        $\states[k]\step{\sact}\states[k+1]$ for all $k\in\N$.
    As before, there exists a state $\statet$ such that
        $\states\brel_1\statet$ and $\statet\brel_2\stateu$.
    From $\states\brel_1\statet$ it follows that there exist an
    infinite sequence of states
      $(\statet[\ell])_{\ell\in\N}$
    and a mapping
      $\statemap:\N\rightarrow\N$
    such that
      $\statet=\statet[0]$,
      $\statet[\ell]\step{\sact}\statet[\ell+1]$
    and
      $\states[\statemap(\ell)]\brel\statet[\ell]$ for all $\ell\in\N$.
    Hence, since $\statet\brel_2\stateu$, it follows that there exist
    an infinite sequence of states
      $(\stateu[m])_{m\in\N}$
    and a mapping
      $\statemap[\rho]:\N\rightarrow\N$
    such that
      $\stateu=\stateu[0]$,
      \plat{$\stateu[m]\step{\sact}\stateu[m+1]$}
    and
      $\statet[{\statemap[\rho](m)}]\brel_2\stateu[m]$ for all $m\in\N$.
    Clearly,
      $\states[\statemap({\statemap[\rho]}(m))]\brel\stateu[m]$
    for all $m\in\N$.
  \qed
  \end{enumerate}
  \end{proof}

  \begin{theorem} \label{theo:wbbisimdeq}
    $\wbbisimd$ is an equivalence.
  \end{theorem}
  \begin{proof}
    The diagonal on $\States$ (i.e., the binary relation
    $\{(\states,\states)\mid\states\in\States\}$) is a symmetric
    relation that satisfies (\ref{cnd:stepsim}) and
    (\ref{cnd:newdivsimshort}), so $\wbbisimd$ is reflexive.
    Furthermore, that $\wbbisimd$ is symmetric is immediate from the
    required symmetry of the witnessing relation.

    To prove that $\wbbisimd$ is transitive, suppose
    $\states\wbbisimd\statet$ and $\statet\wbbisimd\stateu$.
    Then there exist symmetric binary relations $\brelsym_1$ and
    $\brelsym_2$ satisfying (\ref{cnd:stepsim}) and
    (\ref{cnd:newdivsim}) such that
      $\states\brel_1\statet$ and $\statet\brel_2\stateu$.
    The relation
      $\brelsym={(\brel_1\relcomp\brel_2)\cup(\brel_2\relcomp\brel_1)}$
    is clearly symmetric and, by Lemmas \ref{lem:relunion} and
    \ref{lem:relcomp}, satisfies (\ref{cnd:stepsim}) and
    (\ref{cnd:newdivsim}). 
    Hence, since $\states\brel\stateu$, it follows that
    $\states\wbbisimd\stateu$.
  \end{proof}

  \subsection[]{$\wbbisimdsym$ is included in $\rbbisimdsym$}
  \label{subsec:winr}

  To prove the inclusion $\wbbisimdsym\subseteq\rbbisimdsym$ we
  establish that $\wbbisimdsym$ is a branching bisimulation with
  explicit divergence.

  \begin{lemma} \label{lem:wbbisimdlargest}
    The relation $\wbbisimdsym$ satisfies {\rm (\ref{cnd:stepsim})} and
    {\rm (\ref{cnd:newdivsim})}.
  \end{lemma}
  \begin{proof}
    Directly from the definition it follows that $\wbbisimdsym$ is the
    union of all symmetric relations satisfying
    (\ref{cnd:stepsim}) and (\ref{cnd:newdivsim}), so, using
    Lemma~\ref{lem:relunion}, $\wbbisimdsym$ itself satisfies
    (\ref{cnd:stepsim}) and (\ref{cnd:newdivsim}). \vspace{1ex}
  \end{proof}
    \noindent
    In fact, it is now clear that $\wbbisimdsym$ is the largest
    symmetric binary relation satisfying (\ref{cnd:stepsim}) and
    (\ref{cnd:newdivsim}).

  \begin{lemma} \label{lem:wbbisimdrvgdivsim}
    The relation $\wbbisimdsym$ satisfies {\rm (\ref{cnd:rvgdivsim})}.
  \end{lemma}
  \begin{proof}
    Suppose that
      $\states\wbbisimd\statet$
    and that there is an infinite sequence of states
      $(\states[k])_{k\in\N}$
    such that
      $\states=\states[0]$,
      $\states[k]\step{\sact}\states[k+1]$
    and
      $\states[k]\wbbisimd\statet$ for all $k\in\N$.
    According to Lemma~\ref{lem:wbbisimdlargest}, the relation
    $\wbbisimdsym$ satisfies (\ref{cnd:newdivsim}), so there exist an
    infinite sequence of states
      $(\statet[\ell])_{\ell\in\N}$
    and a mapping
      $\statemap:\N\rightarrow\N$
    such that
      $\statet=\statet[0]$,
      $\statet[\ell]\step{\sact}\statet[\ell+1]$
    and
      $\states[\statemap(\ell)]\wbbisimd\statet[\ell]$ for all $\ell\in\N$.
    By Theorem~\ref{theo:wbbisimdeq}, $\wbbisimdsym$ is an
    equivalence, so it follows from
      $\states[k]\wbbisimd\statet$,
      $\states[\statemap(\ell)]\wbbisimd\statet$
    and
      $\states[\statemap(\ell)]\wbbisimd\statet[\ell]$
    that
      $\states[k]\wbbisimd\statet[\ell]$ for all $k,\ell\in\N$.
    Hence $\wbbisimdsym$ satisfies (\ref{cnd:rvgdivsim}).
  \end{proof}

  \begin{theorem} \label{theo:wbbisimdinclrbbisimd}
    $\wbbisimdsym\subseteq\rbbisimdsym$.
  \end{theorem}
  \begin{proof}
    By Theorem~\ref{theo:wbbisimdeq}, the relation $\wbbisimdsym$ is symmetric.
    By Lemma~\ref{lem:wbbisimdlargest}, it satisfies (\ref{cnd:stepsim}), and
    by Lemma~\ref{lem:wbbisimdrvgdivsim} it satisfies
    (\ref{cnd:rvgdivsim}).  So $\wbbisimdsym$ is a branching
    bisimulation with explicit divergence, and hence
    $\states\wbbisimd\statet$ implies $\states\rbbisimd\statet$.
  \end{proof}

\subsection[]{Stuttering closure}

  \begin{definition} \label{df:stuttering}
    A binary relation $\brelsym$
    has the \emph{stuttering property} if, whenever
    $\statet[0]\step{\sact}\cdots\step{\sact}\statet[n]$,
    $\states\brel\statet[0]$ and
    $\states\brel\statet[n]$, then
    $\states\brel\statet[i]$ for all $i=0,\dots,n$.
  \end{definition}

  \noindent
  The following operation converts any binary relation $\brelsym$ on
  $\States$ into a larger relation $\bbdscsym$ that has the stuttering
  property.

  \begin{definition}
    Let $\brelsym$ be a binary relation on $\States$.
    The \emph{stuttering closure} $\bbdscsym$ of $\brelsym$ is defined by
    \begin{equation*}
      \bbdscsym =
        \{(\states,\statet)
             \mid \is{\bstate{s},\astate{s},\bstate{t},\astate{t}\in\States}.\
                    \bstate{s}\ssteps{}\state{s}\ssteps{}\astate{s}\ \&\
                    \bstate{t}\ssteps{}\state{t}\ssteps{}\astate{t}\ \&\
                    \bstate{s}\brel\astate{t}\ \&\
                    \astate{s}\brel\bstate{t}
        \}
    \enskip.
    \end{equation*}
  \end{definition}

\begin{figure}[thb]
\begin{center}\input{stuttering.pstex_t}\end{center}
\caption{Stuttering closure.}
\label{fig:stutclos}
\end{figure}
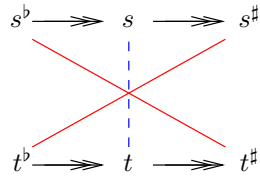

  \noindent
  Figure~\ref{fig:stutclos} illustrates the notion of stuttering
  closure.
  Clearly $\brelsym \subseteq \bbdscsym$.
  We establish a few basic properties of the stuttering closure.

  \begin{lemma} \label{lem:stuttering}
    The stuttering closure of a binary relation has the
    \emph{stuttering property}.
  \end{lemma}
  \begin{proof}
    Let $\brelsym$ be a binary relation and let $\bbdscsym$ be its
    stuttering closure.
    To show that $\bbdscsym$ has the stuttering property, suppose that
      $\statet[0]\step{\sact}\cdots\step{\sact}\statet[n]$,
      $\states\bbdsc\statet[0]$
    and
      $\states\bbdsc\statet[n]$.
    Then, on the one hand, there exist states $\astate{s}$ and
    $\bstate{t_0}$ such that
      $\states\ssteps{}\astate{s}$,
      $\bstate{t_0}\ssteps{}\statet[0]$
    and
      $\astate{s}\brel\bstate{t_0}$,
    and on the other hand there exist states $\bstate{s}$ and
    $\astate{t_n}$ such that
      $\bstate{s}\ssteps{}\states$,
      $\statet[n]\ssteps{}\astate{t_n}$
    and
      $\bstate{s}\brel\astate{t_n}$.
    Now, since
      $\bstate{s}\ssteps{}\states\ssteps{}\astate{s}$
    and
      $\bstate{t_0}\ssteps{}\statet[i]\ssteps{}\astate{t_n}$
       for all $i=0,\dots,n$,
    it follows that $\states\bbdsc\statet[i]$.
  \end{proof}

  \begin{remark}
    The stuttering closure $\bbdscsym$ of a binary relation $\brel$ is
    (contrary to what our terminology may suggest) not necessarily the
    smallest relation containing $\brel$ with the stuttering property.
    For a counterexample, consider a transition system with states
    $\bstate{s}$, $\astate{s}$, $\bstate{t}$ and $\astate{t}$ and
    transitions $\bstate{s}\step{\silent}\astate{s}$ and
    $\bstate{t}\step{\silent}\astate{t}$; the binary relation
     $$\brelsym=\{(\bstate{s},\astate{t}),(\astate{t},\bstate{s}),
                  (\astate{s},\bstate{t}),(\bstate{t},\astate{s}),
                  (\astate{s},\astate{t}),(\astate{t},\astate{s})\}$$
    has the stuttering property, but $\bbdscsym$ has additionally the
    pairs $(\bstate{s},\bstate{t})$ and $(\bstate{t},\bstate{s})$.
  \end{remark}

  \begin{lemma} \label{lem:bbdscsym}
    The stuttering closure $\bbdsc$ of a symmetric binary relation
    $\brelsym$ is symmetric.
  \end{lemma}
  \begin{proof}
    Suppose $\states\bbdsc\statet$; then there exist states
    $\bstate{s}$, $\astate{s}$, $\bstate{t}$ and $\astate{t}$ such
    that
      $\bstate{s}\ssteps{}\states\ssteps{}\astate{s}$,
      $\bstate{t}\ssteps{}\statet\ssteps{}\astate{t}$,
      $\bstate{s}\brel\astate{t}$
    and
      $\astate{s}\brel\bstate{t}$.
    Since $\brel$ is symmetric, it follows that
      $\bstate{t}\brel\astate{s}$
    and
      $\astate{t}\brel\bstate{s}$.
    Hence $\statet\bbdsc\states$.
  \end{proof}

  \begin{lemma} \label{lem:bbdscsync}
    Let $\bbdscsym$ be the stuttering closure of $\brelsym$.
    If $\brelsym$ satisfies {\rm (\ref{cnd:stepsim})} and
    $\states\bbdsc\statet$, then there exists $\state{t'}$ such
    that $\statet\ssteps{}\state{t'}$ and $\state{s}\brel\state{t'}$.
  \end{lemma}
  \begin{proof}
    Suppose $\states\bbdsc\statet$; then there exist states
    $\bstate{s}$, $\astate{s}$, $\bstate{t}$ and $\astate{t}$ such
    that
      $\bstate{s}\ssteps{}\states\ssteps{}\astate{s}$,
      $\bstate{t}\ssteps{}\statet\ssteps{}\astate{t}$,
      $\bstate{s}\brel\astate{t}$
    and
      $\astate{s}\brel\bstate{t}$.
    From $\bstate{s}\brel\astate{t}$ and
    $\bstate{s}\ssteps{}\states$
    it follows by Lemma~\ref{lem:isync} that there exists
    $\state{t'}$ such that
    $(\statet\mathalpha{\ssteps{}})\,\astate{t}\ssteps{}t'$ and
    $\states\brel\state{t'}$.
  \end{proof}

  \begin{lemma} \label{lem:bbdscstepsim}
    If $\brelsym$ satisfies {\rm (\ref{cnd:stepsim})}, then so does
    its stuttering closure $\bbdscsym$.
  \end{lemma}
  \begin{proof}
    Suppose that $\states\bbdsc\statet$ and that
    $\states\step{\acta}\state{s'}$ for some $\state{s'}$.
    Then by Lemma~\ref{lem:bbdscsync} there exists $\state{t^{\dagger}}$
    such that $\statet\ssteps{}\state{t^{\dagger}}$ and
    $\states\brel\state{t^{\dagger}}$.
    Hence, since $\states\step{\acta}\state{s'}$, it follows by
    (\ref{cnd:stepsim}) that there exist states $\state{t''}$ and
    $\state{t'}$ such that
    \begin{equation*}
      (\statet\mathalpha{\ssteps{}})\,
      \state{t^{\dagger}}\ssteps{}\state{t''}\step{\opt{\acta}}\state{t'},\
      \states\brel\state{t''}\
    \text{and}\
      \state{s'}\brel\state{t'}
    \enskip.
    \end{equation*}
    Now, $\states\brel\state{t''}$ and $\state{s'}\brel\state{t'}$
    respectively imply $\states\bbdsc\state{t''}$ and
    $\state{s'}\bbdsc\state{t'}$.
  \end{proof}

  \subsection[]{Closing the cycle of inclusions} \label{subsec:closing}

  Using the notion of stuttering closure we can now prove
  $\bbisimdsym\subseteq\wbbisimdsym$, thereby closing the cycle of
  inclusions. To prove the inclusion we establish that if $\brelsym$
  is a witnessing relation for $\bbisimdsym$, then $\bbdscsym$ is a
  witnessing relation for $\wbbisimdsym$.

  \begin{lemma} \label{lem:bbdscnewdivsim}
    If $\brelsym$ satisfies  {\rm (\ref{cnd:stepsim})} and\
    {\rm (\ref{cnd:divsim})}, then $\bbdscsym$ satisfies
    {\rm (\ref{cnd:newdivsim})}.
  \end{lemma}

  \begin{proof}
    Suppose that $\brelsym$ satisfies (\ref{cnd:stepsim}) and
    (\ref{cnd:divsim}). By Proposition~\ref{prop:short} it suffices to
    establish that $\bbdscsym$ satisfies (\ref{cnd:newdivsimshort}).
    Suppose that $\states\bbdsc\statet$ and there exists an infinite
    sequence of states $(\states[k])_{k\in\N}$ such that
      $\states=\states[0]$
    and
      $\states[k]\step{\sact}\states[k+1]$ for all $k\in\N$.
    We have to show that there exists a state $\state{t'}$ such that
    $\statet\step{\sact}\state{t'}$ and $\states[k]\bbdsc\state{t'}$
    for some $k\in\N$.

    As $\states\bbdsc\statet$, by Lemma~\ref{lem:bbdscsync}
    there exist
      $\statet[0],\dots,\statet[n]$
    such that
      $\statet=\statet[0]
                 \step{\sact}\cdots\step{\sact}
               \statet[n]$
    and
      $\states\brel\statet[n]$.
    By Lemma~\ref{lem:stuttering},
    $\states\bbdsc\statet[i]$ for all $i=0,\dots,n$,
    so if $n>0$, then we can take $\state{t'}=\statet[1]$.
    We proceed with the assumption that $n=0$; so
    $\states\brel\statet$.

    First suppose that $\states[k]\brel\statet$ for all $k\in\N$.
    Then by condition~(\ref{cnd:divsim}) there exist
      $\statet[0],\dots,\statet[m]$
    such that
      $\statet=\statet[0]
                 \step{\sact}\cdots\step{\sact}
               \statet[m]$
    with $m>0$ and $\states[k]\brel\statet[m]$ for some $k\in\N$.
    As $\states[k]\bbdsc\statet[0]$ and $\states[k]\bbdsc\statet[m]$,
    it follows by Lemma~\ref{lem:stuttering} that
    $\states[k]\bbdsc\statet[i]$ for all $i=0,\dots,n$.
    Hence, in particular, $\states[k]\bbdsc\statet[1]$, so we can take
    $\state{t'}=\statet[1]$.

    In the remaining case there is a $k_0\in\N$ such that
    $\states[k]\brel\statet$ for all $k\leq k_0$ while
    $\states[k_0+1]$ and $\state{t}$ are \emph{not} related by
    $\brel$.
    Since $\states[k_0]\brel\statet$ and
    $\states[k_0]\step{\sact}\states[k_0+1]$, by
    condition~(\ref{cnd:stepsim}) of Definition~\ref{def:bbisimd}
    there exist states
      $\statet[0],\dots,\statet[m],\statet[m+1]$
    such that
      \plat{$\statet=\statet[0]
                 \step{\sact}\cdots\step{\sact}
               \statet[m]\step{\opt{\sact}}\statet[m+1]$},
      $\states[k_0]\brel\statet[m]$
    and
      $\states[k_0+1]\brel\statet[m+1]$.
    Since $\states[k_0+1]$ and $\state{t}$ are not related by
    $\brel$, it follows that $\statet[0]\not=\statet[m+1]$,
    so either $m>0$ or $\statet[m]\step{\sact}\statet[m+1]$.
    In case $m>0$, since
      $\states[k_0]\bbdsc\statet[0]$
    and
      $\states[k_0]\bbdsc\statet[m]$,
    by Lemma~\ref{lem:stuttering} it follows that
    $\states[k_0]\bbdsc\statet[1]$, so we can take
    $\state{t'}=\statet[1]$.
    In case $m=0$ and
      $\statet=\statet[m]\step{\sact}\statet[m+1]$,
    we can take $\state{t'}=\statet[m+1]$.
  \end{proof}

  \begin{theorem} \label{theo:bbisimdinclwbbisimd}
     $\bbisimdsym\subseteq\wbbisimdsym$.
  \end{theorem}
  \begin{proof}
    Suppose that $\states\bbisimd\statet$.
    Then there exists a binary relation $\brelsym$ satisfying
    (\ref{cnd:stepsim}) and (\ref{cnd:divsim}),
    such that $\states\brel\statet$.
    By Lemma~\ref{lem:bbdscsym} the stuttering closure $\bbdsc$ of
    $\brelsym$ is symmetric, by Lemma~\ref{lem:bbdscstepsim} it
    satisfies (\ref{cnd:stepsim}), and by
    Lemma~\ref{lem:bbdscnewdivsim} it satisfies (\ref{cnd:newdivsim}).
    Moreover, $\states\bbdsc\statet$.
    Hence, $\states\wbbisimd\statet$.
 \vspace{1ex}
  \end{proof}

  The inclusions already established in Section~\ref{sec:relchar}
  together with the inclusions established in Theorems
  \ref{theo:wbbisimdinclrbbisimd} and \ref{theo:bbisimdinclwbbisimd}
  yield the following corollary (see also Figure~\ref{fig:inclusion}).

  \begin{corollary} \label{cor:bbisimdeqwbbisimd}
    $\rbbisimdsym=\bbisimdsym=\wbbisimdsym$.
  \qed\end{corollary}

  \begin{corollary} \label{cor:bbisimdeq}
    The relation $\rbbisimdsym$ is an equivalence.
  \qed\end{corollary}

  \noindent
  Recall that the proof strategy employed in
  Lemma~\ref{lem:relunion}(\ref{lem:relunion:item:newdivsim})
  to show that any union of binary relations satisfying
  (\ref{cnd:newdivsim}) also satisfies (\ref{cnd:newdivsim}),
  fails with (\ref{cnd:rvgdivsim}) or (\ref{cnd:divsim}) instead of
  (\ref{cnd:newdivsim}).
  In fact, it is easy to show that these results do not even hold.
  Therefore, we could not directly infer from the definition of
  $\rbbisimdsym$ that it is itself a branching bisimulation with
  explicit divergence. But now it follows, for
    $\rbbisimdsym=\wbbisimdsym$ satisfies (\ref{cnd:stepsim}) and
  (\ref{cnd:newdivsim}) by Lemma~\ref{lem:wbbisimdlargest}, and
  hence also the weaker condition (\ref{cnd:divsim}). It  satisfies
  (\ref{cnd:rvgdivsim}) by Lemma~\ref{lem:wbbisimdrvgdivsim}.

  \begin{corollary} \label{cor:bbisimdbrel}
    $\rbbisimdsym$ is the largest symmetric relation satisfying
    {\rm (\ref{cnd:stepsim})} and {\rm (\ref{cnd:divsim})}. It even
    satisfies {\rm (\ref{cnd:rvgdivsim})}, {\rm (\ref{cnd:newdivsim})}
    and {\rm (\ref{cnd:newdivsimshort})}. It therefore is the largest
    branching bisimulation with explicit divergence.
  \qed\end{corollary}

  \noindent
  The following corollary is another consequence, which we need in the
  next section.

  \begin{corollary} \label{cor:bbisimdstuttering}
    The relation $\rbbisimdsym$ has the \emph{stuttering property}.
  \end{corollary}
  \begin{proof}
    Since $\rbbisimdsym$ satisfies (\ref{cnd:stepsim})
    and (\ref{cnd:divsim}), its stuttering
    closure $\widehat\bisim_b^{\Delta}$ satisfies (\ref{cnd:stepsim})
    and  (\ref{cnd:newdivsim}) by Lemmas~\ref{lem:bbdscstepsim}
    and~\ref{lem:bbdscnewdivsim}.
    Moreover, $\widehat\bisim_b^{\Delta}$ is symmetric by
    Lemma~\ref{lem:bbdscsym}. Therefore
    $\widehat\bisim_b^{\Delta}$ is included in $\wbbisimdsym$ (cf.\
    the proof of Lemma~\ref{lem:wbbisimdlargest}). As
    $\rbbisimdsym\subseteq\widehat\bisim_b^{\Delta}\subseteq\wbbisimdsym$
    we find $\rbbisimdsym=\widehat\bisim_b^{\Delta}$.
    Thus, by Lemma~\ref{lem:stuttering}, $\rbbisimdsym$ has the
    stuttering property.
  \end{proof}

\section[Coloured-trace characterisation]
        {Coloured-trace characterisation of $\bbd$} \label{sec:ctchar}

  We now recall from \cite{GW96} the original characterisation in
  terms of coloured traces of branching bisimilarity with
  explicit divergence, and establish that it coincides with the
  relational characterisations of Section~\ref{sec:relchar}.

  \begin{definition}\label{def:colouringsext}
    Let $\C$ be a colouring.
    A state $s$ is \emph{$\C$-divergent} if there exists an infinite
    sequence of states
        $(\states[k])_{k\in\N}$
      such that
        $\states=\states[0]$,
        $\states[k]\step{\sact}\states[k+1]$
      and
        $\C(\states[k])=\C(s)$ for all $k\in\N$.
    A consistent colouring is said to \emph{preserve divergence} if no
    $\C$-divergent state has the same colour as a nondivergent state.

    We write $s\ccd t$ if there exists a consistent, divergence
    preserving colouring $\C$ with $\C(s)=\C(t)$.
  \end{definition}

\noindent
  We prove that $\ccdsym = \rbbisimdsym$.

\begin{lemma}\label{BB2}
Let $\C$ be a colouring such that two states with the same colour have
the same $\C$-coloured traces of length three (i.e.\ colour - action - colour).
Then $\C$ is consistent.
\end{lemma}

\begin{proof}
Suppose $\C(s_0)=\C(t_0)$ and $C_0, a_1, C_1, \dots, a_n, C_n$ is a
coloured trace of $s_0$. Then, for $i=1,\dots,n$, there are states $s_i$
and paths $\pi_i$ from $s_{i-1}$ to $s_{i}$, such that
$\C(\pi_i)=C_{i-1},a_i,C_{i}$. With induction on $i$, for $i=1,\dots,n$
we find states $t_i$ with $C(s_i)=C(t_i)$ and paths $\rho_i$
from $t_{i-1}$ to $t_{i}$ such that $\C(\rho_i)=C_{i-1},a_i,C_{i}$.
Namely, the assumption about $\C$ allows us to find $\rho_i$ given
$t_{i-1}$, and then $t_i$ is defined as the last state of $\rho_i$.
Concatenating all the paths $\rho_i$ yields a path $\rho$
from $t_0$ with $\C(\rho)=C_0, a_1, C_1,\dots, a_n, C_n$.
\end{proof}

\begin{theorem}
$\ccdsym = \rbbisimdsym$.
\end{theorem}

\begin{proof}
``$\subseteq$'':
Let $\C$ be a consistent colouring that preserves divergence.
It suffices to show that $\C$ is a branching bisimulation with explicit
divergence.

Suppose $s ~\C~ t$, i.e.\ $\C(s)=\C(t)$, and $s \step{a}s'$ for some
state $s'$. In case $a=\tau$ and $\C(s')=\C(s)$ we have $s' ~\C~ t$
and condition (\ref{cnd:stepsim}) is satisfied. So suppose $a\neq\tau$
or $\C(s')\neq\C(s)$. Then $s$, and therefore also $t$, has a coloured
trace $\C(s),a,\C(s')$. This implies that there are states
$t_0,\dots,t_n$ for some $n\geq 0$ and $t'$ with $t=t_0\step\tau t_1
\step\tau \cdots \step\tau t_n \step{\opt{a}} t'$ such that
$\C(t_i)=\C(s)$ for $i=0,...,n$ and $\C(t')=\C(s')$. Hence
(\ref{cnd:stepsim}) is satisfied.

Now suppose $s ~\C~ t$ and there is an infinite sequence of states
      $(\states[k])_{k\in\N}$
    such that
      $\states=\states[0]$,
      $\states[k]\step{\sact}\states[k+1]$
    and
      $\states[k]~\C~\statet$ for all $k\in\N$.
Then $\C(\states[k])=\C(s)$ for all $k\in\N$.
Hence $s$, and therefore also $t$, is $\C$-divergent.
Thus there exists an infinite sequence of states
      $(\statet[\ell])_{\ell\in\N}$
    such that
      $\statet=\statet[0]$,
      $\statet[\ell]\step{\sact}\statet[\ell+1]$
    and
      $\C(\statet[\ell])=\C(t)$ for all $\ell\in\N$.
It follows that
      $\states[k]~\C~\statet[\ell]$ for all $k,\ell\in\N$.
Hence also (\ref{cnd:rvgdivsim}) is satisfied.

``$\supseteq$'': It suffices to show that $\rbbisimdsym$ is a
consistent, divergence preserving colouring.  By
Corollary~\ref{cor:bbisimdeq} it is an equivalence.  We also
use that it satisfies (\ref{cnd:stepsim}) and (\ref{cnd:rvgdivsim})
(Corollary~\ref{cor:bbisimdbrel}) and
has the stuttering property (Corollary~\ref{cor:bbisimdstuttering}).
We invoke Lemma~\ref{BB2} for proving consistency.

Suppose that $\states$ and $\statet$ have the same colour, i.e.,
$s\rbbisimd t$, and let $C,\acta,D$ be a $\rbbisimdsym$-coloured trace
of $s$. Then $\acta\neq\sact$ or $C\neq D$, and there are states
$\state{s''}$ and $\state{s'}$ with
$\states\ssteps\state{s''}\step{\acta}\state{s'}$, such that
$\state{s''},\state{s}\in C$ and $\state{s'}\in D$. As $\rbbisimdsym$
satisfies (\ref{cnd:stepsim}), by Lemma~\ref{lem:isync} there is a
state $\state{t^\dagger}$ with $\statet\ssteps\state{t^\dagger}$ and
$\state{s''}\rbbisimd \state{t^\dagger}$.
Therefore there exist states $\state{t''}$ and $\state{t'}$ such that
  $(\statet \ssteps)\,
      \state{t^\dagger}\ssteps\state{t''}\step{\opt{\acta}}\state{t'}$,
  $\state{s''}\rbbisimd\state{t''}$ and
  $\state{s'}\rbbisimd{}\state{t'}$.
As $\rbbisimdsym$ has the stuttering property and
$\state{t''}\rbbisimd\state{s''}\rbbisimd\states\rbbisimd\statet$,
all states on the $\sact$-path from $t$ to $t''$ have the same colour
as $\states$. Hence  $C,\acta,D$ is a $\rbbisimdsym$-coloured trace of
$\statet$.

Now suppose $\states$ and $\statet$ have the same colour and $\states$
is $\rbbisimd$-divergent. Then there is an infinite sequence of states
      $(\states[k])_{k\in\N}$
such that
      $\states=\states[0]$,
      $\states[k]\step{\sact}\states[k+1]$
    and
      $\states[k]\rbbisimd\states\rbbisimd\statet$ for all $k\in\N$.
As $\rbbisimd$ satisfies (\ref{cnd:rvgdivsim}), this implies
that there exists an infinite sequence of states
      $(\statet[\ell])_{\ell\in\N}$
    such that
      $\statet=\statet[0]$,
      $\statet[\ell]\step{\sact}\statet[\ell+1]$
    and
      $\states[k]\rbbisimd\statet[\ell]$ for all $k,\ell\in\N$.
It follows that
      $\statet[\ell]\rbbisimd \statet$ for all $\ell\in\N$,
and hence $\statet$ is $\rbbisimd$-divergent.
\end{proof}

  \section[Modal characterisations]
          {Modal characterisations of $\bbd$} \label{sec:modchar}

  We shall now establish agreement between the relational and modal
  characterisations of $\bbd$ proposed in \cite{Gla93}. The class of
  formulas $\JBEDFrm$ of the modal logic for $\bbd$ proposed in
  \cite{Gla93} is generated by the grammar obtained by adding the following
  clause to the grammar in \eqref{eq:grammar} of Section~\ref{sec:bb}:
  \begin{equation} \label{eq:grammarext}
    \frm\ ::=\
      \rDiv\frm
        \qquad\text{($\varphi\in\JBEDFrm$).}
  \end{equation}
  We extend the inductive definition of validity in
  Section~\ref{sec:bb} with:
  \begin{enumerate}\itemsep 0pt
  \renewcommand{\theenumi}{\roman{enumi}}
  \renewcommand{\labelenumi}{(\theenumi)}
  \addtocounter{enumi}{3}
  \item
    $\states\sat\rDiv\frm$
      \IFF{}
    there exists an infinite sequence $(\states[k])_{k\in\omega}$
    of states such that $\states\ssteps\states[0]$,
      $\states[k]\step{\sact}\states[k+1]$
    and
      $\states[k]\sat\frm$ for all $k\in\omega$.
  \end{enumerate}
  Again, validity induces an equivalence on states: we define
  $\requivalidsym\subseteq\States\times\States$ by
  \begin{equation*}
    \states\requivalid\statet\quad
  \text{iff}\quad
    \all{\frm\in\JBEDFrm}.\
      {\states\sat\frm}\Leftrightarrow{\statet\sat\frm}
  \enskip.
  \end{equation*}
  We shall now establish that $\requivalidsym$ coincides with
  $\rbbisimd$.

  \begin{theorem} \label{theo:rmodchar}
    For all states $\states$ and $\statet$:
      $\states\rbbisimd\statet$
    \IFF{}
      $\states\requivalid\statet$.
  \end{theorem}
  \begin{proof}
    To establish the implication from left to right,
    we prove by structural induction on $\frm$ that if
    $\states\rbbisimd\statet$ and $\states\sat\frm$, then
    $\statet\sat\frm$. 
    There are four cases to consider.
    \begin{enumerate}
    \item Suppose $\frm=\neg\psi$ and $\states\sat\frm$.
      Then $\states\not\sat\psi$. As $\statet\rbbisimd\states$, it
      follows by the induction hypothesis that $\statet\not\sat\psi$,
      and hence $\statet\sat\frm$.
    \item Suppose $\states\sat\bigwedge \Psi$. Then, for all $\psi\in\Psi$,
      $\states\sat\psi$, and by induction $\statet\sat\psi$. This
      yields $\statet\sat\phi$.
    \item Suppose $\frm=\frm[\psi]\runtil[\acta]\frm[\chi]$ and
      $\states\sat\frm$.
      Then there exist states $\state{s'}$ and $\state{s''}$ such that
        \plat{$\states\ssteps\state{s''}\step{\opt{\acta}}\state{s'}$},
        $\state{s''}\sat\frm[\psi]$ and $\state{s'}\sat\frm[\chi]$.
      By Lemma~\ref{lem:isync}, there exists a state
      $\state{t^\dagger}$ such that
        \plat{$\statet\ssteps\state{t^\dagger}$}
      and
        $\state{s''}\rbbisimd\state{t^\dagger}$.
      From this it follows that there exist states $\state{t'}$ and
      $\state{t''}$ such that
        \plat{$\statet\ssteps\state{t''}\step{\opt{\acta}}\state{t'}$},
        $\state{s''}\rbbisimd\state{t''}$ and
        $\state{s'}\rbbisimd\state{t'}$,
      for if $\acta=\silent$ and $\state{s'}=\state{s''}$ we can take
      $\state{t'}=\state{t''}=\state{t^{\dagger}}$ and otherwise,
      since $\state{s''}\rbbisimd\state{t^\dagger}$, the states
      $\state{t'}$ and $\state{t''}$ exist by (\ref{cnd:stepsim}).
      It follows by the induction hypothesis that
        $\state{t''}\sat\frm[\psi]$ and $\state{t'}\sat\frm[\chi]$,
      and hence $\statet\sat\frm$.
    \item Suppose $\frm=\rDiv\frm[\psi]$ and
      $\states\sat\frm$. Then there exists an infinite sequence
      $(\states[k])_{k\in\N}$ such that $\states\ssteps\states[0]$,
      $\states[k]\step{\sact}\states[k+1]$ and
      $\states[k]\sat\frm[\psi]$ for all $k\in\N$.
      By Lemma~\ref{lem:isync}, there exists a state
      $\statet[0]$ such that
        \plat{$\statet\ssteps\statet[0]$}
      and
        $\states[0]\rbbisimd\statet[0]$.
      From Corollary~\ref{cor:bbisimdbrel} it follows that
      $\rbbisimdsym$ satisfies (\ref{cnd:newdivsim}), so
      there exist an infinite sequence of
      states $(\statet[\ell])_{\ell\in\N}$ and a mapping
        $\statemap:\N\rightarrow\N$
      such that
        $\statet[\ell]\step{\sact}\statet[\ell+1]$
      and
        $\states[\statemap(\ell)]\rbbisimd\statet[\ell]$ for all
        $\ell\in\N$.
      By the induction hypothesis $\statet[\ell]\sat\frm[\psi]$ for all
      $\ell\in\N$, and hence $\statet\sat\frm$.
    \end{enumerate}
    For the implication from right to left, it suffices by
    Corollary~\ref{cor:bbisimdeqwbbisimd} to prove that
    $\requivalidsym$ is symmetric and satisfies the
    conditions~(\ref{cnd:stepsim}) and (\ref{cnd:divsim}).

    That $\requivalidsym$ is symmetric is clear from its definition.

    To establish condition~(\ref{cnd:stepsim}) of
    Definition~\ref{def:bbisimd}, suppose that
    $\states\requivalid\statet$ and $\states\step{\acta}\state{s'}$,
    and define sets $T''$ and $T'$ as follows:
    \begin{gather*}
      \States[T'']=\{\state{t''}\in\States
            \mid\statet\ssteps{}\state{t''}\ \&\
            \states\not\requivalid\state{t''}
          \};\ \text{and}\\
      \States[T']=\{\state{t'}\in\States
            \mid\exists{\state{t''}\in\States}.\
                  \statet\ssteps{}\state{t''}\step{\opt{\acta}}\state{t'}\
                \&\
                  \state{s'}\not\requivalid\state{t'}
         \}
    \;.
    \end{gather*}
    For each $\state{t''}\in\States[T'']$ let
    $\frm[{\varphi_{\state{t''}}}]$
    be a formula such that
      $\states\sat\frm[{\varphi_{\state{t''}}}]$ and
      $\state{t''}\not\sat\frm[{\varphi_{\state{t''}}}]$,
    and let
      $\frm[\varphi]
         =\Meet\{\frm[\varphi_{\state{t''}}]\mid\state{t''}\in
         \States[T'']\}$.
    Similarly, for each $\state{t'}\in\States[T']$ let
    $\frm[{\psi_{\state{t'}}}]$ be a formula with
    $\state{s'}\sat\frm[{\psi_{\state{t'}}}]$ and
    $\state{t'}\not\sat\frm[{\psi_{\state{t'}}}]$,
    and let
      $\frm[\psi]
         =\Meet\{\frm[\psi_{\state{t'}}]
                   \mid \state{t'}\in\States[T']\}$.
    Note that $\states\sat\frm[\varphi]\runtil[\acta]\frm[\psi]$, and
    hence, since $\states\requivalid\statet$, also
      $\statet\sat\frm[\varphi]\runtil[\acta]\frm[\psi]$.
    So, there exist states $\state{t'}$ and $\state{t''}$ such that
      \plat{$\statet\ssteps\state{t''}\step{\opt{\acta}}\state{t'}$},
      $\state{t''}\sat\frm[\varphi]$ and $\state{t'}\sat\frm[\psi]$.
    From $\state{t''}\sat\frm[\varphi]$ it follows that
      $\state{t''}\not\in\States[T'']$,
    so $\states\requivalid\state{t''}$;
    from $\state{t'}\sat\frm[\psi]$ it follows that
      $\state{t'}\not\in\States[T']$,
    so $\state{s'}\requivalid\state{t'}$.
    Thereby, condition~(\ref{cnd:stepsim}) is established.

    To establish condition~(\ref{cnd:divsim}), suppose that
    $\states\requivalid\statet$ and that there exists an
    infinite sequence $(\states[k])_{k\in\N}$
    such that $\states=\states[0]$,
    $\states[k]\step{\sact}\states[k+1]$ and
    $\states[k]\requivalid\statet$ for all $k\in\N$.
    Define the set $\States[T^\infty]$ by
    \begin{equation*}
      \States[T^\infty]
        = \{\state{t'}\in\States\mid
              \statet\ssteps{}\state{t'}\ \&\
              \states\not\requivalid\state{t'}\}
    \;.
    \end{equation*}
    For each $t'\in\States[T^\infty]$ let $\frm[\varphi_{\state{t'}}]$
    be a formula such that $\states\sat\frm[\varphi_{\state{t'}}]$ and
      $\state{t'}\not\sat\frm[\varphi_{t'}]$,
    and let
      $\frm=
         \Meet\{\frm[\varphi_{\state{t'}}]
                  \mid\state{t'}\in\States[T^\infty]\}$.
    Since $\states[0]=\states\sat\frm$ and
      $\states[k]\requivalid\statet\requivalid\states$,
    it follows that $\states[k]\sat\frm$ for
    all $k\in\N$, and therefore $\states\sat\rDiv\frm$.
    Hence, $\statet\sat\rDiv\frm$, so there exists an infinite
    sequence $(\statet[\ell])_{\ell\in\N}$ such that
    $\statet\ssteps\statet[0]$,
    $\statet[\ell]\step{\sact}\statet[\ell+1]$ and
    $\statet[\ell]\sat\frm$ for all $\ell\in\N$.
    It follows that $\statet[\ell]\not\in\States[T^\infty]$,
    so $\states\requivalid\statet[\ell]$, for all $\ell\in\N$, and
    hence $\states[k]\requivalid\states\requivalid\statet[\ell]$ for all
    $k,\ell\in\N$.
    It follows, in particular, that
      \plat{$\statet\plussteps\statet[1]$}
      and
        $\states[k]\requivalid\statet[1]$ for some $k\in\N$.
    Thereby, also condition~(\ref{cnd:divsim}) is established.\vspace{1ex}
  \end{proof}

  We already mentioned in Section~\ref{sec:bb} the result of
  Laroussinie, Pinchinat \& Schnoebelen \cite{LPS95} that the modal
  logic with negation, binary conjunction and $\funtil$ and the logic
  with negation, binary conjunction and $\runtil$ are equally
  expressive. Below, we adapt their method to show that replacing
  $\runtil$ by $\funtil$ or $\until$ in the modal logic for $\bbd$
  proposed in \cite{Gla93} also yields an equally expressive logic.

  Henceforth we denote by $\UEDFrm$ the set of formulas generated by
  the grammar that is obtained when replacing $\frm\runtil\frm$ by
  $\frm\until\frm$ in the grammar for $\JBEDFrm$ (see
  \eqref{eq:grammar} in Section~\ref{sec:bb} and \eqref{eq:grammarext}
  at the beginning of this section). The central idea, from
  \cite{LPS95}, is that any formula in $\JBEDFrm$ can be written as a
  Boolean combination of formulas that propagate either upwards or
  downwards along a path of $\tau$-transitions.  A formula $\frm$ that
  propagates upwards, i.e., with the property that if $s\ssteps t$ and
  $s\sat\frm$, then also $t\sat\frm$, we shall call an \emph{upward
    formula}. A formula $\frm$ that propagates downwards, i.e., with
  the property that if $s\ssteps t$ and $t\sat\frm$, then also
  $s\sat\frm$, we shall call a \emph{downward formula}.

  \begin{lemma} \label{lem:separation}
    Every $\frm\in\JBEDFrm$ is equivalent with a
    formula of the form $\bigvee\Frm[\Phi]$, where each formula in
    $\Frm[\Phi]$ is a conjunction of an upward and a downward formula.
  \end{lemma}
  \begin{proof}
    Note that $\frm[\psi]\runtil\frm[\chi]$ and $\rDiv\frm[\psi]$ are
    downward formulas and that the negation of a downward formula is
    an upward formula. Furthermore, a conjunction of upward formulas is
    again an upward formula and a conjunction of downward formulas is
    again a downward formula. It follows, by the standard laws of
    Boolean algebra, that the formula $\frm$ is equivalent to a formula
    of the desired shape.
\vspace{1ex}  \end{proof}

\noindent
  The proof that for every formula $\frm \in \UEDFrm$ there exists an
  equivalent formula $\frm[\varphi']\in\JBEDFrm$ proceeds by
  induction on the structure of $\frm$, and the only nontrivial case
  is when $\frm=\frm[\psi]\until\frm[\chi]$. According to the
  induction hypothesis, for $\frm[\psi]$ and $\frm[\chi]$ there exist
  equivalent formulas in $\JBEDFrm$, so, by Lemma~\ref{lem:separation},
  $\frm[\psi]$ is equivalent to a disjunction of conjunctions of
  upward and downward formulas. The proof in \cite{LPS95} then relies
  on these disjunctions being finite. To generalise it to infinite
  disjunctions, we shall use the following lemma.

  \begin{lemma} \label{lem:infdisjunctions}
    Let $\Frm$ be a set of formulas and let $\frm$ be a formula.
    Then
    \begin{alignat*}{3}
      & \left(\textstyle\bigvee\Frm\right) \until \frm
      &&\ ~\equi~\
        \textstyle\bigvee\left\{
                       \left(\textstyle\bigvee\frm[\Phi']\right)\until\frm
                     \mid \text{$\frm[\Phi']$ a finite subset of $\Frm$}
                   \right\}
    \enskip.
    \end{alignat*}
  \end{lemma}
  \begin{proof}
  \begin{enumerate}
  \item[($\Rrightarrow$)]
    Suppose
      $\states\sat\left(\bigvee\Frm\right) \until[\acta] \frm$.
    Then there exist states $\states[0],\dots,\states[n],\states[n+1]$
    such that
      $\states=\states[0]
                 \step{\sact}\cdots\step{\sact}
               \states[n]\step{\opt{\acta}}\states[n+1]$,
    $\states[i]\sat\bigvee\Frm$ for all $i=0,\dots,n$ and
    $\states[n+1]\sat\frm$.
    Since $\states[i]\sat\bigvee\Frm$, we can associate with
    every $\states[i]$ ($i=0,\dots,n$) a formula
    $\frm[\varphi_i]\in\Frm$ such that $\states[i]\sat\frm[\varphi_i]$.
    Let
      $\Frm[\Phi']=\{\frm[\varphi_i]\mid i=0,\dots,n\}$;
    then $\Frm[\Phi']$ is a finite subset of $\Frm$ such that
      $\states[i]\sat\bigvee\Frm[\Phi']$ for every $i=0,\dots,n$.
    It follows that
      $\states\sat(\bigvee\frm[\Phi'])\until\frm$,
    and hence
      $\states\sat\bigvee\{
                      (\bigvee\Frm[\Phi'])\until\frm
                    \mid \text{$\Frm[\Phi']$ a finite subset of $\Frm$}
                  \}$.
  \item[($\Lleftarrow$)]
    If $\states\sat\bigvee\{
                      (\bigvee\Frm[\Phi'])\until\frm
                    \mid \text{$\Frm[\Phi']$ a finite subset of $\Frm$}
                   \}$,
    then $\states\sat(\bigvee\Frm[\Phi'])\until\frm$ for some finite subset
    $\Frm[\Phi']$ of $\Frm$. So there exist states
    $\states[0],\dots,\states[n],\states[n+1]$ such that
      $\states=\states[0]
                 \step{\sact}\cdots\step{\sact}
               \states[n]\step{\opt{\acta}}\states[n+1]$,
    $\states[i]\sat\bigvee\Frm[\Phi']$ for all $i=0,\dots,n$ and
    $\states[n+1]\sat\frm$. Since $\states[i]\sat\bigvee\Frm[\Phi']$ implies
    $\states[i]\sat\bigvee\Frm$ for all $i=0,\dots,n$, it follows that
      $\states\sat(\bigvee\Frm)\until\frm$.
  \qed
  \end{enumerate}
  \end{proof}

\noindent
  We now adapt the method in \cite{LPS95} and show that replacing
  $\runtil$ by $\funtil$ or $\until$ in the modal logic for $\bbd$
  proposed in \cite{Gla93} yields an equally expressive logic.

  \begin{theorem}
    For every formula $\frm\in\UEDFrm$ there exists an
    equivalent formula
    $\frm[\varphi']\in\JBEDFrm$.
  \end{theorem}
  \begin{proof}
    The proof is by structural induction on $\frm$; the only
    nontrivial case is when $\frm=\frm[\psi]\until\frm[\chi]$.
    By the induction hypothesis there exist formulas
    $\frm[\psi'],\frm[\chi']\in\JBEDFrm$ such that
    $\frm[\psi]\equi\frm[\psi']$
    and
    $\frm[\chi]\equi\frm[\chi']$.
    By Lemma~\ref{lem:separation}, $\frm[\psi']\equi\bigvee\Frm[\Psi]$,
    where each formula in $\Frm[\Psi]$ is a conjunction of an
    upward and a downward formula. Hence, by the evident congruence
    property of $\equi$ and Lemma~\ref{lem:infdisjunctions},
    \begin{alignat*}{3}
      & \frm
      &&\ ~\equi~\
        \textstyle\bigvee\left\{
                       \left(\textstyle\bigvee\Psi'\right)\until\chi'
                     \mid \text{$\Frm[\Psi']$ a finite subset of
                       $\Frm[\Psi]$}
                   \right\}
    \enskip.
    \end{alignat*}
    Clearly, it now suffices to establish that
      $(\bigvee\Frm[\Psi'])\until\frm[\chi']$
    is equivalent to a formula in $\JBEDFrm$, for all finite subsets
    $\Frm[\Psi']$ of $\Frm[\Psi]$.
    Recall that $\Frm[\Psi]$ consists of conjunctions of an upward and
    a downward formula, so we can assume that
      $\Frm[\Psi']=
         \{\frm[\psi_i^\mathrm{u}]\meet\frm[\psi_i^\mathrm{d}]\mid
             i=1,\dots,n\}$;
    we proceed by induction on the cardinality of $\Frm[\Psi']$.

    \noindent
    If $|\Frm[\Psi']|=0$, then
    \begin{alignat*}{3}
      & \left(\bigvee\Frm[\Psi']\right)\until\frm[\chi']
      &&\ ~\equi~\
        \false
    \enskip,
    \end{alignat*}
    and $\false\in\JBEDFrm$.

    \noindent
    Suppose $|\Frm[\Psi']|>0$.
    By the induction hypothesis there exists, for every $i=1,\dots,n$,
    a formula $\frm[\varphi_i']\in \JBEDFrm$ such that
    \begin{alignat*}{3}
      &  \left(\bigvee\Frm[\Psi']-
                          \{\frm[\psi_i^\mathrm{u}]
                              \meet
                            \frm[\psi_i^\mathrm{d}]
                          \}
          \right)\until\frm[\chi']
      &&\ ~\equi~\
        \frm[\varphi_i']
    \enskip.
    \end{alignat*}
    Then, it is easy to verify that
    \begin{alignat*}{3}
      & \left(\bigvee\Frm[\Psi']\right)\until\frm[\chi']
      &&\ ~\equi~\
        \bigvee_{i=1}^n\left(
                     \frm[\psi_i^\mathrm{u}]
                       \meet
                     \left(\frm[\psi_i^\mathrm{d}]\runtil\frm[\chi']
                             \vee
                             \frm[\psi_i^\mathrm{d}]\runtil[\sact]
                               \frm[\varphi_i']
                     \right)
                    \right)
    \enskip,
    \end{alignat*}
    and the right-hand side formula is in $\JBEDFrm$.
    Some intuition for this last step is offered in \cite{LPS95}.
\vspace{1ex}  \end{proof}

\noindent
  In the same vain, there is also an obvious strengthening of the
  divergence modality $\rDiv$. Let $\Div$ be the unary divergence
  modality with the following definition:
  \begin{enumerate}
  \item[(iv$'$)] $\states\sat\Div\frm$
          \IFF{}
        there exists an infinite sequence $(\states[k])_{k\in\omega}$
        of states such that $\states=\states[0]$,
          $\states[k]\step{\sact}\states[k+1]$ and
        $\states[k]\sat\frm$ for all $k\in\omega$.
  \end{enumerate}
  We denote by $\JBSDFrm$ the set of formulas generated by the grammar
  in \eqref{eq:grammar} with $\rDiv\frm$ replaced by $\Div\frm$.

\begin{figure}[htb]
\begin{center}\input{divergence.pstex_t}\end{center}
\caption{A divergence.}
\label{fig:divform}
\end{figure}
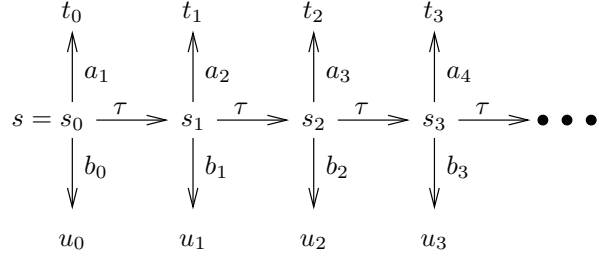

  Note that the modality $\rDiv$ can be expressed in terms of
  $\Div$:
  \begin{alignat*}{3}
    & \rDiv\frm
    &&\ ~\equi~\
      \top\runtil[\sact]\Div\frm\enskip.
  \end{alignat*}

  A crucial step in our adaptation of the method of Laroussinie,
  Pinchinat \& Schnoebelen above consisted of showing that infinite
  disjunctions in the left argument of $\until$ can be avoided. If
  infinite disjunctions could also be avoided as an argument of
  $\Div$, then a further adaptation of the method would be possible,
  showing that replacing $\rDiv$ by $\Div$ in the modal logic for
  $\bbd$ would yield a logic with equal expressivity. However, the
  following example suggests that infinite disjunctions under $\Div$
  cannot always be avoided.

  \begin{example} \label{exa:disjunctiondivergence}
    Let $\acta[1],\acta[2],\acta[3],\dots$ and
    $\actb[0],\actb[1],\actb[2],\dots$ be infinite sequences of
    distinct actions and consider the formula
    \begin{equation*}
    \frm=
      \Div\left(\bigvee_{i=0}^\infty\left(
                  \compl\left(\true\until[{\acta[i]}]\true\right)
                    \meet
                  \left(\true\until[{\actb[i]}]\true\right)
                \right)
      \right)
    \enskip.
    \end{equation*}
    The formula $\frm$ holds in a state iff there exists an infinite
    $\sact$-path such that in every state there is an
    $i\geq 0$ such that the action $\actb[i]$ is still possible, whereas the
    action $\acta[i]$ is not. Note that $\frm$ holds in the state
    $\states$ of the transition system in Figure~\ref{fig:divform};
    each of the disjuncts
      $\compl\left(\true\until[{\acta[i]}]\true\right)
         \meet
       \left(\true\until[{\actb[i]}]\true\right)$
    holds in precisely one state.
\vspace{1ex}  \end{example}

\noindent
  We conjecture that the formula of
  Example~\ref{exa:disjunctiondivergence} is not equivalent to a
  formula in \plat{$\JBEDFrm$}, and that, hence, replacing $\rDiv$ by $\Div$
  in the modal logic for $\bbd$ yields a strictly more expressive
  logic. We conclude the paper with a proof that the
  equivalence
    \plat{$\equivalidsym\subseteq\States\times\States$}
  induced on states by validity of formulas in \plat{$\JBSDFrm$}, defined by
  \begin{equation*}
    \states\equivalid\statet\quad
  \text{iff}\quad
    \all{\frm\in\JBSDFrm}.\
      {\states\sat\frm}\Leftrightarrow{\statet\sat\frm}
  \enskip,
  \end{equation*}
  nevertheless also coincides with $\rbbisimd$.
  \begin{theorem} \label{theo:modchar}
    For all states $\states$ and $\statet$:
      $\states\rbbisimd\statet$
    \IFF{}
      $\states\equivalid\statet$.
  \end{theorem}
  \begin{proof}
    For the implication from left to right, we prove by structural
    induction on $\frm$ that if $\states\rbbisimd\statet$ and
    $\states\sat\frm$, then $\statet\sat\frm$. We only treat
    the case $\frm=\Div\frm[\psi]$, for the cases
    $\frm=\neg\psi$, $\frm=\bigwedge\Psi$ and 
    $\frm=\frm[\psi]\runtil\frm[\chi]$ are already treated in the proof
    of Theorem~\ref{theo:rmodchar}.
    So, suppose $\frm=\Div\frm[\psi]$ and $\states\sat\frm$. Then
    there exists an infinite sequence $(\states[k])_{k\in\omega}$ of
    states such that $\states=\states[0]$,
    $\states[k]\step{\sact}\states[k+1]$ and $\states[k]\sat\frm[\psi]$ for
    all $k\in\omega$. From Corollary~\ref{cor:bbisimdbrel} it follows
    that $\rbbisimdsym$ satisfies (\ref{cnd:newdivsim}), so there
    exist an infinite sequence of states $(\statet[\ell])_{\ell\in\N}$ and a
    mapping $\statemap:\N\rightarrow\N$ such that
      $\statet=\statet[0]$,
      $\statet[\ell]\step{\sact}\statet[\ell+1]$
    and
      $\states[\statemap(\ell)]\rbbisimd\statet[\ell]$ for all $\ell\in\N$.
    By the induction hypothesis $\statet[\ell]\sat\frm[\psi]$ for all
    $\ell\in\N$, and hence $\statet\sat\frm$.

    To establish the implication from right to left, note that if
    $\states\equivalid\statet$, then, since every formula in
    \plat{$\JBEDFrm$} is equivalent to a formula in \plat{$\JBSDFrm$}, also
    $\states\requivalid\statet$, so by Theorem~\ref{theo:rmodchar} it
    follows that $\states\rbbisimd\statet$.
\end{proof}

\paragraph*{Comment on Definition~\ref{expressiveness}}
If in Definition~\ref{expressiveness} we had used a notion of
equivalence between modal formulas $\varphi$ and $\psi$ that merely
requires that $s\sat\varphi \Leftrightarrow s\sat\psi$ for all states
$s$ in the presupposed labelled transition system, rather than
quantifying over all labelled transition systems, the resulting
concept of equally expressive logics would be much weaker, and the
logics \plat{$\JBEDFrm$} and \plat{$\JBSDFrm$} would be equally expressive.

In general, let $\sim$ be an equivalence on the set of states
$S$, and consider two logics $\mathcal{L}_1$ and $\mathcal{L}_2$ that
both have negation and arbitrary infinite conjunction, and both
characterise $\sim$.  For every pair of states $s,t \in S$ with
$s\not\sim t$ take a formula $\varphi_{s,t}$ from $\mathcal{L}_1$
such that $s\sat \varphi_{s,t}$ but $t\not\sat\varphi_{s,t}$. Then
$\chi_s = \bigwedge\{ \varphi_{s,t}\mid t \not\sim s \}$ is called a
\emph{characteristic formula} of $s$: one has $t\sat \chi_s$ iff
$t\sim s$. Now let $\psi$ be a formula from $\mathcal{L}_2$. Then
$\bigvee\{\chi_s\mid s\sat\psi\}$ is equivalent to $\psi$, in the sense
that $t\sat\psi \Leftrightarrow
t\sat\bigvee\{\chi_s\mid s\sat\psi\}$ for all states $t\in S$.
This proves that the two logics are equally expressive.

Similar reasoning using the notion of equivalence from
Definition~\ref{expressiveness} would break down, because one cannot
take conjunctions of a proper class of formula.

\bibliographystyle{plain}
\def\sortunder#1{}

\end{document}

%% file: rvgdivsim.pstex_t
\begin{picture}(0,0)%
\includegraphics{rvgdivsim.pstex}%
\end{picture}%
\setlength{\unitlength}{3947sp}%
\begingroup\makeatletter\ifx\SetFigFont\undefined%
\gdef\SetFigFont#1#2#3#4#5{%
  \reset@font\fontsize{#1}{#2pt}%
  \fontfamily{#3}\fontseries{#4}\fontshape{#5}%
  \selectfont}%
\fi\endgroup%
\begin{picture}(3425,1158)(289,-421)
\put(2251,614){\makebox(0,0)[b]{\smash{{\SetFigFont{10}{12.0}{\familydefault}{\mddefault}{\updefault}{\color[rgb]{0,0,0}$\tau$}%
}}}}
\put(1351,614){\makebox(0,0)[b]{\smash{{\SetFigFont{10}{12.0}{\familydefault}{\mddefault}{\updefault}{\color[rgb]{0,0,0}$\tau$}%
}}}}
\put(3001,614){\makebox(0,0)[b]{\smash{{\SetFigFont{10}{12.0}{\familydefault}{\mddefault}{\updefault}{\color[rgb]{0,0,0}$\tau$}%
}}}}
\put(1351,-286){\makebox(0,0)[b]{\smash{{\SetFigFont{10}{12.0}{\familydefault}{\mddefault}{\updefault}{\color[rgb]{0,.56,0}$\tau$}%
}}}}
\put(601,-286){\makebox(0,0)[b]{\smash{{\SetFigFont{10}{12.0}{\familydefault}{\mddefault}{\updefault}{\color[rgb]{0,.56,0}$\tau$}%
}}}}
\put(3001,-286){\makebox(0,0)[b]{\smash{{\SetFigFont{10}{12.0}{\familydefault}{\mddefault}{\updefault}{\color[rgb]{0,.56,0}$\tau$}%
}}}}
\put(2251,-286){\makebox(0,0)[b]{\smash{{\SetFigFont{10}{12.0}{\familydefault}{\mddefault}{\updefault}{\color[rgb]{0,.56,0}$\tau$}%
}}}}
\put(1051,539){\makebox(0,0)[b]{\smash{{\SetFigFont{10}{12.0}{\familydefault}{\mddefault}{\updefault}{\color[rgb]{0,0,0}$s_1$}%
}}}}
\put(1051,-361){\makebox(0,0)[b]{\smash{{\SetFigFont{10}{12.0}{\familydefault}{\mddefault}{\updefault}{\color[rgb]{0,.56,0}$t_1$}%
}}}}
\put(601,614){\makebox(0,0)[b]{\smash{{\SetFigFont{10}{12.0}{\familydefault}{\mddefault}{\updefault}{\color[rgb]{0,0,0}$\tau$}%
}}}}
\put(376,539){\makebox(0,0)[rb]{\smash{{\SetFigFont{10}{12.0}{\familydefault}{\mddefault}{\updefault}{\color[rgb]{0,0,0}$s=s_0$}%
}}}}
\put(376,-361){\makebox(0,0)[rb]{\smash{{\SetFigFont{10}{12.0}{\familydefault}{\mddefault}{\updefault}{\color[rgb]{0,0,0}$t=t_0$}%
}}}}
\put(2701,539){\makebox(0,0)[b]{\smash{{\SetFigFont{10}{12.0}{\familydefault}{\mddefault}{\updefault}{\color[rgb]{0,0,0}$s_k$}%
}}}}
\put(2701,-361){\makebox(0,0)[b]{\smash{{\SetFigFont{10}{12.0}{\familydefault}{\mddefault}{\updefault}{\color[rgb]{0,.56,0}$t_\ell$}%
}}}}
\end{picture}%

%% file: incompositional_original.pstex_t
\begin{picture}(0,0)%
\includegraphics{incompositional_original.pstex}%
\end{picture}%
\setlength{\unitlength}{3947sp}%
\begingroup\makeatletter\ifx\SetFigFont\undefined%
\gdef\SetFigFont#1#2#3#4#5{%
  \reset@font\fontsize{#1}{#2pt}%
  \fontfamily{#3}\fontseries{#4}\fontshape{#5}%
  \selectfont}%
\fi\endgroup%
\begin{picture}(5355,1533)(286,-745)
\put(1051,-661){\makebox(0,0)[b]{\smash{{\SetFigFont{10}{12.0}{\familydefault}{\mddefault}{\updefault}{\color[rgb]{0,0,0}$u_1$}%
}}}}
\put(1801,-661){\makebox(0,0)[b]{\smash{{\SetFigFont{10}{12.0}{\familydefault}{\mddefault}{\updefault}{\color[rgb]{0,0,0}$u_2$}%
}}}}
\put(1351,-586){\makebox(0,0)[b]{\smash{{\SetFigFont{10}{12.0}{\familydefault}{\mddefault}{\updefault}{\color[rgb]{0,0,0}$\tau$}%
}}}}
\put(601,-586){\makebox(0,0)[b]{\smash{{\SetFigFont{10}{12.0}{\familydefault}{\mddefault}{\updefault}{\color[rgb]{0,0,0}$\tau$}%
}}}}
\put(2476,-661){\makebox(0,0)[lb]{\smash{{\SetFigFont{10}{12.0}{\familydefault}{\mddefault}{\updefault}{\color[rgb]{0,0,0}$\tau$}%
}}}}
\put(301,-661){\makebox(0,0)[b]{\smash{{\SetFigFont{10}{12.0}{\familydefault}{\mddefault}{\updefault}{\color[rgb]{0,0,0}$u_0$}%
}}}}
\put(1051,-61){\makebox(0,0)[b]{\smash{{\SetFigFont{10}{12.0}{\familydefault}{\mddefault}{\updefault}{\color[rgb]{0,0,0}$t_1$}%
}}}}
\put(1801,-61){\makebox(0,0)[b]{\smash{{\SetFigFont{10}{12.0}{\familydefault}{\mddefault}{\updefault}{\color[rgb]{0,0,0}$t_2$}%
}}}}
\put(1351, 14){\makebox(0,0)[b]{\smash{{\SetFigFont{10}{12.0}{\familydefault}{\mddefault}{\updefault}{\color[rgb]{0,0,0}$\tau$}%
}}}}
\put(601, 14){\makebox(0,0)[b]{\smash{{\SetFigFont{10}{12.0}{\familydefault}{\mddefault}{\updefault}{\color[rgb]{0,0,0}$\tau$}%
}}}}
\put(2476,-61){\makebox(0,0)[lb]{\smash{{\SetFigFont{10}{12.0}{\familydefault}{\mddefault}{\updefault}{\color[rgb]{0,0,0}$\tau$}%
}}}}
\put(1051,539){\makebox(0,0)[b]{\smash{{\SetFigFont{10}{12.0}{\familydefault}{\mddefault}{\updefault}{\color[rgb]{0,0,0}$s_1$}%
}}}}
\put(601,614){\makebox(0,0)[b]{\smash{{\SetFigFont{10}{12.0}{\familydefault}{\mddefault}{\updefault}{\color[rgb]{0,0,0}$\tau$}%
}}}}
\put(1801,539){\makebox(0,0)[b]{\smash{{\SetFigFont{10}{12.0}{\familydefault}{\mddefault}{\updefault}{\color[rgb]{0,0,0}$s_2$}%
}}}}
\put(1351,614){\makebox(0,0)[b]{\smash{{\SetFigFont{10}{12.0}{\familydefault}{\mddefault}{\updefault}{\color[rgb]{0,0,0}$\tau$}%
}}}}
\put(2476,539){\makebox(0,0)[lb]{\smash{{\SetFigFont{10}{12.0}{\familydefault}{\mddefault}{\updefault}{\color[rgb]{0,0,0}$\tau$}%
}}}}
\put(301,539){\makebox(0,0)[b]{\smash{{\SetFigFont{10}{12.0}{\familydefault}{\mddefault}{\updefault}{\color[rgb]{0,0,0}$s_0$}%
}}}}
\put(4201,539){\makebox(0,0)[b]{\smash{{\SetFigFont{10}{12.0}{\familydefault}{\mddefault}{\updefault}{\color[rgb]{0,0,0}$s_1$}%
}}}}
\put(4951,539){\makebox(0,0)[b]{\smash{{\SetFigFont{10}{12.0}{\familydefault}{\mddefault}{\updefault}{\color[rgb]{0,0,0}$s_2$}%
}}}}
\put(4501,614){\makebox(0,0)[b]{\smash{{\SetFigFont{10}{12.0}{\familydefault}{\mddefault}{\updefault}{\color[rgb]{0,0,0}$\tau$}%
}}}}
\put(3751,614){\makebox(0,0)[b]{\smash{{\SetFigFont{10}{12.0}{\familydefault}{\mddefault}{\updefault}{\color[rgb]{0,0,0}$\tau$}%
}}}}
\put(5626,539){\makebox(0,0)[lb]{\smash{{\SetFigFont{10}{12.0}{\familydefault}{\mddefault}{\updefault}{\color[rgb]{0,0,0}$\tau$}%
}}}}
\put(3451,539){\makebox(0,0)[b]{\smash{{\SetFigFont{10}{12.0}{\familydefault}{\mddefault}{\updefault}{\color[rgb]{0,0,0}$s_0$}%
}}}}
\put(4201,-661){\makebox(0,0)[b]{\smash{{\SetFigFont{10}{12.0}{\familydefault}{\mddefault}{\updefault}{\color[rgb]{0,0,0}$u_1$}%
}}}}
\put(4951,-661){\makebox(0,0)[b]{\smash{{\SetFigFont{10}{12.0}{\familydefault}{\mddefault}{\updefault}{\color[rgb]{0,0,0}$u_2$}%
}}}}
\put(4501,-586){\makebox(0,0)[b]{\smash{{\SetFigFont{10}{12.0}{\familydefault}{\mddefault}{\updefault}{\color[rgb]{0,0,0}$\tau$}%
}}}}
\put(3751,-586){\makebox(0,0)[b]{\smash{{\SetFigFont{10}{12.0}{\familydefault}{\mddefault}{\updefault}{\color[rgb]{0,0,0}$\tau$}%
}}}}
\put(5626,-661){\makebox(0,0)[lb]{\smash{{\SetFigFont{10}{12.0}{\familydefault}{\mddefault}{\updefault}{\color[rgb]{0,0,0}$\tau$}%
}}}}
\put(3451,-661){\makebox(0,0)[b]{\smash{{\SetFigFont{10}{12.0}{\familydefault}{\mddefault}{\updefault}{\color[rgb]{0,0,0}$u_0$}%
}}}}
\put(301,-61){\makebox(0,0)[b]{\smash{{\SetFigFont{10}{12.0}{\familydefault}{\mddefault}{\updefault}{\color[rgb]{0,0,0}$t_0$}%
}}}}
\end{picture}%

%% file: rvgdivsimshort.pstex_t
\begin{picture}(0,0)%
\includegraphics{rvgdivsimshort.pstex}%
\end{picture}%
\setlength{\unitlength}{3947sp}%
\begingroup\makeatletter\ifx\SetFigFont\undefined%
\gdef\SetFigFont#1#2#3#4#5{%
  \reset@font\fontsize{#1}{#2pt}%
  \fontfamily{#3}\fontseries{#4}\fontshape{#5}%
  \selectfont}%
\fi\endgroup%
\begin{picture}(3427,1153)(286,-416)
\put(2251,614){\makebox(0,0)[b]{\smash{{\SetFigFont{10}{12.0}{\familydefault}{\mddefault}{\updefault}{\color[rgb]{0,0,0}$\tau$}%
}}}}
\put(1351,614){\makebox(0,0)[b]{\smash{{\SetFigFont{10}{12.0}{\familydefault}{\mddefault}{\updefault}{\color[rgb]{0,0,0}$\tau$}%
}}}}
\put(3001,614){\makebox(0,0)[b]{\smash{{\SetFigFont{10}{12.0}{\familydefault}{\mddefault}{\updefault}{\color[rgb]{0,0,0}$\tau$}%
}}}}
\put(601,-286){\makebox(0,0)[b]{\smash{{\SetFigFont{10}{12.0}{\familydefault}{\mddefault}{\updefault}{\color[rgb]{0,.56,0}$\tau$}%
}}}}
\put(1051,539){\makebox(0,0)[b]{\smash{{\SetFigFont{10}{12.0}{\familydefault}{\mddefault}{\updefault}{\color[rgb]{0,0,0}$s_1$}%
}}}}
\put(1051,-361){\makebox(0,0)[b]{\smash{{\SetFigFont{10}{12.0}{\familydefault}{\mddefault}{\updefault}{\color[rgb]{0,.56,0}$t'$}%
}}}}
\put(601,614){\makebox(0,0)[b]{\smash{{\SetFigFont{10}{12.0}{\familydefault}{\mddefault}{\updefault}{\color[rgb]{0,0,0}$\tau$}%
}}}}
\put(376,539){\makebox(0,0)[rb]{\smash{{\SetFigFont{10}{12.0}{\familydefault}{\mddefault}{\updefault}{\color[rgb]{0,0,0}$s=s_0$}%
}}}}
\put(2701,539){\makebox(0,0)[b]{\smash{{\SetFigFont{10}{12.0}{\familydefault}{\mddefault}{\updefault}{\color[rgb]{0,0,0}$s_k$}%
}}}}
\put(301,-361){\makebox(0,0)[b]{\smash{{\SetFigFont{10}{12.0}{\familydefault}{\mddefault}{\updefault}{\color[rgb]{0,0,0}$t$}%
}}}}
\end{picture}%

%% file: onestepdivsim.pstex_t
\begin{picture}(0,0)%
\includegraphics{onestepdivsim.pstex}%
\end{picture}%
\setlength{\unitlength}{3947sp}%
\begingroup\makeatletter\ifx\SetFigFont\undefined%
\gdef\SetFigFont#1#2#3#4#5{%
  \reset@font\fontsize{#1}{#2pt}%
  \fontfamily{#3}\fontseries{#4}\fontshape{#5}%
  \selectfont}%
\fi\endgroup%
\begin{picture}(3427,1153)(286,-416)
\put(2251,614){\makebox(0,0)[b]{\smash{{\SetFigFont{10}{12.0}{\familydefault}{\mddefault}{\updefault}{\color[rgb]{0,0,0}$\tau$}%
}}}}
\put(1351,614){\makebox(0,0)[b]{\smash{{\SetFigFont{10}{12.0}{\familydefault}{\mddefault}{\updefault}{\color[rgb]{0,0,0}$\tau$}%
}}}}
\put(3001,614){\makebox(0,0)[b]{\smash{{\SetFigFont{10}{12.0}{\familydefault}{\mddefault}{\updefault}{\color[rgb]{0,0,0}$\tau$}%
}}}}
\put(601,-286){\makebox(0,0)[b]{\smash{{\SetFigFont{10}{12.0}{\familydefault}{\mddefault}{\updefault}{\color[rgb]{0,.56,0}$\tau$}%
}}}}
\put(1051,539){\makebox(0,0)[b]{\smash{{\SetFigFont{10}{12.0}{\familydefault}{\mddefault}{\updefault}{\color[rgb]{0,0,0}$s_1$}%
}}}}
\put(1051,-361){\makebox(0,0)[b]{\smash{{\SetFigFont{10}{12.0}{\familydefault}{\mddefault}{\updefault}{\color[rgb]{0,.56,0}$t'$}%
}}}}
\put(601,614){\makebox(0,0)[b]{\smash{{\SetFigFont{10}{12.0}{\familydefault}{\mddefault}{\updefault}{\color[rgb]{0,0,0}$\tau$}%
}}}}
\put(376,539){\makebox(0,0)[rb]{\smash{{\SetFigFont{10}{12.0}{\familydefault}{\mddefault}{\updefault}{\color[rgb]{0,0,0}$s=s_0$}%
}}}}
\put(2701,539){\makebox(0,0)[b]{\smash{{\SetFigFont{10}{12.0}{\familydefault}{\mddefault}{\updefault}{\color[rgb]{0,0,0}$s_k$}%
}}}}
\put(301,-361){\makebox(0,0)[b]{\smash{{\SetFigFont{10}{12.0}{\familydefault}{\mddefault}{\updefault}{\color[rgb]{0,0,0}$t$}%
}}}}
\end{picture}%

%% file: incompositional_onestepdivsim.pstex_t
\begin{picture}(0,0)%
\includegraphics{incompositional_onestepdivsim.pstex}%
\end{picture}%
\setlength{\unitlength}{3947sp}%
\begingroup\makeatletter\ifx\SetFigFont\undefined%
\gdef\SetFigFont#1#2#3#4#5{%
  \reset@font\fontsize{#1}{#2pt}%
  \fontfamily{#3}\fontseries{#4}\fontshape{#5}%
  \selectfont}%
\fi\endgroup%
\begin{picture}(5355,2503)(286,-1691)
\put(2476,389){\makebox(0,0)[lb]{\smash{{\SetFigFont{10}{12.0}{\familydefault}{\mddefault}{\updefault}{\color[rgb]{0,0,0}$\tau$}%
}}}}
\put(5626,389){\makebox(0,0)[lb]{\smash{{\SetFigFont{10}{12.0}{\familydefault}{\mddefault}{\updefault}{\color[rgb]{0,0,0}$\tau$}%
}}}}
\put(1426,-811){\makebox(0,0)[b]{\smash{{\SetFigFont{10}{12.0}{\familydefault}{\mddefault}{\updefault}{\color[rgb]{0,0,0}$\tau$}%
}}}}
\put(1051,-511){\makebox(0,0)[b]{\smash{{\SetFigFont{10}{12.0}{\familydefault}{\mddefault}{\updefault}{\color[rgb]{0,0,0}$t_1$}%
}}}}
\put(1801,-511){\makebox(0,0)[b]{\smash{{\SetFigFont{10}{12.0}{\familydefault}{\mddefault}{\updefault}{\color[rgb]{0,0,0}$t_2$}%
}}}}
\put(1351,-361){\makebox(0,0)[b]{\smash{{\SetFigFont{10}{12.0}{\familydefault}{\mddefault}{\updefault}{\color[rgb]{0,0,0}$\tau$}%
}}}}
\put(2101,-361){\makebox(0,0)[b]{\smash{{\SetFigFont{10}{12.0}{\familydefault}{\mddefault}{\updefault}{\color[rgb]{0,0,0}$\tau$}%
}}}}
\put(601,-361){\makebox(0,0)[b]{\smash{{\SetFigFont{10}{12.0}{\familydefault}{\mddefault}{\updefault}{\color[rgb]{0,0,0}$\tau$}%
}}}}
\put(301,-511){\makebox(0,0)[b]{\smash{{\SetFigFont{10}{12.0}{\familydefault}{\mddefault}{\updefault}{\color[rgb]{0,0,0}$t_0$}%
}}}}
\put(2551,-511){\makebox(0,0)[b]{\smash{{\SetFigFont{10}{12.0}{\familydefault}{\mddefault}{\updefault}{\color[rgb]{0,0,0}$t_3$}%
}}}}
\put(1051,389){\makebox(0,0)[b]{\smash{{\SetFigFont{10}{12.0}{\familydefault}{\mddefault}{\updefault}{\color[rgb]{0,0,0}$s_0$}%
}}}}
\put(301,389){\makebox(0,0)[b]{\smash{{\SetFigFont{10}{12.0}{\familydefault}{\mddefault}{\updefault}{\color[rgb]{0,0,0}$s_1$}%
}}}}
\put(1801,389){\makebox(0,0)[b]{\smash{{\SetFigFont{10}{12.0}{\familydefault}{\mddefault}{\updefault}{\color[rgb]{0,0,0}$s_2$}%
}}}}
\put(601,689){\makebox(0,0)[b]{\smash{{\SetFigFont{10}{12.0}{\familydefault}{\mddefault}{\updefault}{\color[rgb]{0,0,0}$\tau$}%
}}}}
\put(1351,539){\makebox(0,0)[b]{\smash{{\SetFigFont{10}{12.0}{\familydefault}{\mddefault}{\updefault}{\color[rgb]{0,0,0}$\tau$}%
}}}}
\put(301,-1411){\makebox(0,0)[b]{\smash{{\SetFigFont{10}{12.0}{\familydefault}{\mddefault}{\updefault}{\color[rgb]{0,0,0}$u_0$}%
}}}}
\put(1051,-1411){\makebox(0,0)[b]{\smash{{\SetFigFont{10}{12.0}{\familydefault}{\mddefault}{\updefault}{\color[rgb]{0,0,0}$u_1$}%
}}}}
\put(1801,-1411){\makebox(0,0)[b]{\smash{{\SetFigFont{10}{12.0}{\familydefault}{\mddefault}{\updefault}{\color[rgb]{0,0,0}$u_2$}%
}}}}
\put(1051,-1636){\makebox(0,0)[b]{\smash{{\SetFigFont{10}{12.0}{\familydefault}{\mddefault}{\updefault}{\color[rgb]{0,0,0}$\tau$}%
}}}}
\put(4201,389){\makebox(0,0)[b]{\smash{{\SetFigFont{10}{12.0}{\familydefault}{\mddefault}{\updefault}{\color[rgb]{0,0,0}$s_0$}%
}}}}
\put(3451,389){\makebox(0,0)[b]{\smash{{\SetFigFont{10}{12.0}{\familydefault}{\mddefault}{\updefault}{\color[rgb]{0,0,0}$s_1$}%
}}}}
\put(4951,389){\makebox(0,0)[b]{\smash{{\SetFigFont{10}{12.0}{\familydefault}{\mddefault}{\updefault}{\color[rgb]{0,0,0}$s_2$}%
}}}}
\put(3751,689){\makebox(0,0)[b]{\smash{{\SetFigFont{10}{12.0}{\familydefault}{\mddefault}{\updefault}{\color[rgb]{0,0,0}$\tau$}%
}}}}
\put(4501,539){\makebox(0,0)[b]{\smash{{\SetFigFont{10}{12.0}{\familydefault}{\mddefault}{\updefault}{\color[rgb]{0,0,0}$\tau$}%
}}}}
\put(3451,-1411){\makebox(0,0)[b]{\smash{{\SetFigFont{10}{12.0}{\familydefault}{\mddefault}{\updefault}{\color[rgb]{0,0,0}$u_0$}%
}}}}
\put(4201,-1411){\makebox(0,0)[b]{\smash{{\SetFigFont{10}{12.0}{\familydefault}{\mddefault}{\updefault}{\color[rgb]{0,0,0}$u_1$}%
}}}}
\put(4951,-1411){\makebox(0,0)[b]{\smash{{\SetFigFont{10}{12.0}{\familydefault}{\mddefault}{\updefault}{\color[rgb]{0,0,0}$u_2$}%
}}}}
\put(4201,-1636){\makebox(0,0)[b]{\smash{{\SetFigFont{10}{12.0}{\familydefault}{\mddefault}{\updefault}{\color[rgb]{0,0,0}$\tau$}%
}}}}
\put(601,-1261){\makebox(0,0)[b]{\smash{{\SetFigFont{10}{12.0}{\familydefault}{\mddefault}{\updefault}{\color[rgb]{0,0,0}$\tau$}%
}}}}
\put(1351,-1261){\makebox(0,0)[b]{\smash{{\SetFigFont{10}{12.0}{\familydefault}{\mddefault}{\updefault}{\color[rgb]{0,0,0}$\tau$}%
}}}}
\put(751,164){\makebox(0,0)[b]{\smash{{\SetFigFont{10}{12.0}{\familydefault}{\mddefault}{\updefault}{\color[rgb]{0,0,0}$\tau$}%
}}}}
\put(3901,164){\makebox(0,0)[b]{\smash{{\SetFigFont{10}{12.0}{\familydefault}{\mddefault}{\updefault}{\color[rgb]{0,0,0}$\tau$}%
}}}}
\put(3751,-1261){\makebox(0,0)[b]{\smash{{\SetFigFont{10}{12.0}{\familydefault}{\mddefault}{\updefault}{\color[rgb]{0,0,0}$\tau$}%
}}}}
\put(4501,-1261){\makebox(0,0)[b]{\smash{{\SetFigFont{10}{12.0}{\familydefault}{\mddefault}{\updefault}{\color[rgb]{0,0,0}$\tau$}%
}}}}
\end{picture}%

%% file: newdivsimshort.pstex_t
\begin{picture}(0,0)%
\includegraphics{newdivsimshort.pstex}%
\end{picture}%
\setlength{\unitlength}{3947sp}%
\begingroup\makeatletter\ifx\SetFigFont\undefined%
\gdef\SetFigFont#1#2#3#4#5{%
  \reset@font\fontsize{#1}{#2pt}%
  \fontfamily{#3}\fontseries{#4}\fontshape{#5}%
  \selectfont}%
\fi\endgroup%
\begin{picture}(3427,1153)(286,-416)
\put(2251,614){\makebox(0,0)[b]{\smash{{\SetFigFont{10}{12.0}{\familydefault}{\mddefault}{\updefault}{\color[rgb]{0,0,0}$\tau$}%
}}}}
\put(1351,614){\makebox(0,0)[b]{\smash{{\SetFigFont{10}{12.0}{\familydefault}{\mddefault}{\updefault}{\color[rgb]{0,0,0}$\tau$}%
}}}}
\put(3001,614){\makebox(0,0)[b]{\smash{{\SetFigFont{10}{12.0}{\familydefault}{\mddefault}{\updefault}{\color[rgb]{0,0,0}$\tau$}%
}}}}
\put(601,-286){\makebox(0,0)[b]{\smash{{\SetFigFont{10}{12.0}{\familydefault}{\mddefault}{\updefault}{\color[rgb]{0,.56,0}$\tau$}%
}}}}
\put(1051,539){\makebox(0,0)[b]{\smash{{\SetFigFont{10}{12.0}{\familydefault}{\mddefault}{\updefault}{\color[rgb]{0,0,0}$s_1$}%
}}}}
\put(1051,-361){\makebox(0,0)[b]{\smash{{\SetFigFont{10}{12.0}{\familydefault}{\mddefault}{\updefault}{\color[rgb]{0,.56,0}$t'$}%
}}}}
\put(601,614){\makebox(0,0)[b]{\smash{{\SetFigFont{10}{12.0}{\familydefault}{\mddefault}{\updefault}{\color[rgb]{0,0,0}$\tau$}%
}}}}
\put(376,539){\makebox(0,0)[rb]{\smash{{\SetFigFont{10}{12.0}{\familydefault}{\mddefault}{\updefault}{\color[rgb]{0,0,0}$s=s_0$}%
}}}}
\put(2701,539){\makebox(0,0)[b]{\smash{{\SetFigFont{10}{12.0}{\familydefault}{\mddefault}{\updefault}{\color[rgb]{0,0,0}$s_k$}%
}}}}
\put(301,-361){\makebox(0,0)[b]{\smash{{\SetFigFont{10}{12.0}{\familydefault}{\mddefault}{\updefault}{\color[rgb]{0,0,0}$t$}%
}}}}
\end{picture}%

%% file: newdivsim.pstex_t
\begin{picture}(0,0)%
\includegraphics{newdivsim.pstex}%
\end{picture}%
\setlength{\unitlength}{3947sp}%
\begingroup\makeatletter\ifx\SetFigFont\undefined%
\gdef\SetFigFont#1#2#3#4#5{%
  \reset@font\fontsize{#1}{#2pt}%
  \fontfamily{#3}\fontseries{#4}\fontshape{#5}%
  \selectfont}%
\fi\endgroup%
\begin{picture}(3425,1158)(289,-421)
\put(2251,614){\makebox(0,0)[b]{\smash{{\SetFigFont{10}{12.0}{\familydefault}{\mddefault}{\updefault}{\color[rgb]{0,0,0}$\tau$}%
}}}}
\put(1351,614){\makebox(0,0)[b]{\smash{{\SetFigFont{10}{12.0}{\familydefault}{\mddefault}{\updefault}{\color[rgb]{0,0,0}$\tau$}%
}}}}
\put(3001,614){\makebox(0,0)[b]{\smash{{\SetFigFont{10}{12.0}{\familydefault}{\mddefault}{\updefault}{\color[rgb]{0,0,0}$\tau$}%
}}}}
\put(1351,-286){\makebox(0,0)[b]{\smash{{\SetFigFont{10}{12.0}{\familydefault}{\mddefault}{\updefault}{\color[rgb]{0,.56,0}$\tau$}%
}}}}
\put(601,-286){\makebox(0,0)[b]{\smash{{\SetFigFont{10}{12.0}{\familydefault}{\mddefault}{\updefault}{\color[rgb]{0,.56,0}$\tau$}%
}}}}
\put(3001,-286){\makebox(0,0)[b]{\smash{{\SetFigFont{10}{12.0}{\familydefault}{\mddefault}{\updefault}{\color[rgb]{0,.56,0}$\tau$}%
}}}}
\put(2251,-286){\makebox(0,0)[b]{\smash{{\SetFigFont{10}{12.0}{\familydefault}{\mddefault}{\updefault}{\color[rgb]{0,.56,0}$\tau$}%
}}}}
\put(1051,539){\makebox(0,0)[b]{\smash{{\SetFigFont{10}{12.0}{\familydefault}{\mddefault}{\updefault}{\color[rgb]{0,0,0}$s_1$}%
}}}}
\put(1051,-361){\makebox(0,0)[b]{\smash{{\SetFigFont{10}{12.0}{\familydefault}{\mddefault}{\updefault}{\color[rgb]{0,.56,0}$t_1$}%
}}}}
\put(601,614){\makebox(0,0)[b]{\smash{{\SetFigFont{10}{12.0}{\familydefault}{\mddefault}{\updefault}{\color[rgb]{0,0,0}$\tau$}%
}}}}
\put(376,539){\makebox(0,0)[rb]{\smash{{\SetFigFont{10}{12.0}{\familydefault}{\mddefault}{\updefault}{\color[rgb]{0,0,0}$s=s_0$}%
}}}}
\put(376,-361){\makebox(0,0)[rb]{\smash{{\SetFigFont{10}{12.0}{\familydefault}{\mddefault}{\updefault}{\color[rgb]{0,0,0}$t=t_0$}%
}}}}
\put(2701,539){\makebox(0,0)[b]{\smash{{\SetFigFont{10}{12.0}{\familydefault}{\mddefault}{\updefault}{\color[rgb]{0,0,0}$s_k$}%
}}}}
\put(2701,-361){\makebox(0,0)[b]{\smash{{\SetFigFont{10}{12.0}{\familydefault}{\mddefault}{\updefault}{\color[rgb]{0,.56,0}$t_\ell$}%
}}}}
\end{picture}%

%% file: divsim.pstex_t
\begin{picture}(0,0)%
\includegraphics{divsim.pstex}%
\end{picture}%
\setlength{\unitlength}{3947sp}%
\begingroup\makeatletter\ifx\SetFigFont\undefined%
\gdef\SetFigFont#1#2#3#4#5{%
  \reset@font\fontsize{#1}{#2pt}%
  \fontfamily{#3}\fontseries{#4}\fontshape{#5}%
  \selectfont}%
\fi\endgroup%
\begin{picture}(3424,1158)(289,-421)
\put(2251,614){\makebox(0,0)[b]{\smash{{\SetFigFont{10}{12.0}{\familydefault}{\mddefault}{\updefault}{\color[rgb]{0,0,0}$\tau$}%
}}}}
\put(1351,614){\makebox(0,0)[b]{\smash{{\SetFigFont{10}{12.0}{\familydefault}{\mddefault}{\updefault}{\color[rgb]{0,0,0}$\tau$}%
}}}}
\put(3001,614){\makebox(0,0)[b]{\smash{{\SetFigFont{10}{12.0}{\familydefault}{\mddefault}{\updefault}{\color[rgb]{0,0,0}$\tau$}%
}}}}
\put(601,-286){\makebox(0,0)[b]{\smash{{\SetFigFont{10}{12.0}{\familydefault}{\mddefault}{\updefault}{\color[rgb]{0,.56,0}$\tau$}%
}}}}
\put(2251,-286){\makebox(0,0)[b]{\smash{{\SetFigFont{10}{12.0}{\familydefault}{\mddefault}{\updefault}{\color[rgb]{0,.56,0}$\tau$}%
}}}}
\put(1351,-286){\makebox(0,0)[b]{\smash{{\SetFigFont{10}{12.0}{\familydefault}{\mddefault}{\updefault}{\color[rgb]{0,.56,0}$\tau$}%
}}}}
\put(1051,539){\makebox(0,0)[b]{\smash{{\SetFigFont{10}{12.0}{\familydefault}{\mddefault}{\updefault}{\color[rgb]{0,0,0}$s_1$}%
}}}}
\put(601,614){\makebox(0,0)[b]{\smash{{\SetFigFont{10}{12.0}{\familydefault}{\mddefault}{\updefault}{\color[rgb]{0,0,0}$\tau$}%
}}}}
\put(376,539){\makebox(0,0)[rb]{\smash{{\SetFigFont{10}{12.0}{\familydefault}{\mddefault}{\updefault}{\color[rgb]{0,0,0}$s=s_0$}%
}}}}
\put(376,-361){\makebox(0,0)[rb]{\smash{{\SetFigFont{10}{12.0}{\familydefault}{\mddefault}{\updefault}{\color[rgb]{0,0,0}$t=t_0$}%
}}}}
\put(2701,539){\makebox(0,0)[b]{\smash{{\SetFigFont{10}{12.0}{\familydefault}{\mddefault}{\updefault}{\color[rgb]{0,0,0}$s_k$}%
}}}}
\put(2701,-361){\makebox(0,0)[b]{\smash{{\SetFigFont{10}{12.0}{\familydefault}{\mddefault}{\updefault}{\color[rgb]{0,.56,0}$t'$}%
}}}}
\put(1051,-361){\makebox(0,0)[b]{\smash{{\SetFigFont{10}{12.0}{\familydefault}{\mddefault}{\updefault}{\color[rgb]{0,.56,0}$t_1$}%
}}}}
\end{picture}%

%% file: inclusiongraph.pstex_t
\begin{picture}(0,0)%
\includegraphics{inclusiongraph.pstex}%
\end{picture}%
\setlength{\unitlength}{3947sp}%
\begingroup\makeatletter\ifx\SetFigFont\undefined%
\gdef\SetFigFont#1#2#3#4#5{%
  \reset@font\fontsize{#1}{#2pt}%
  \fontfamily{#3}\fontseries{#4}\fontshape{#5}%
  \selectfont}%
\fi\endgroup%
\begin{picture}(1867,1758)(699,-1546)
\put(2551,-1411){\makebox(0,0)[b]{\smash{{\SetFigFont{10}{12.0}{\familydefault}{\mddefault}{\updefault}{\color[rgb]{0,0,0}$\rbbisimd$}%
}}}}
\put(751,-1411){\makebox(0,0)[b]{\smash{{\SetFigFont{10}{12.0}{\familydefault}{\mddefault}{\updefault}{\color[rgb]{0,0,0}$\wbbisimd$}%
}}}}
\put(1651,-661){\makebox(0,0)[b]{\smash{{\SetFigFont{10}{12.0}{\familydefault}{\mddefault}{\updefault}{\color[rgb]{0,0,0}$\onestepbbisimd$}%
}}}}
\put(1651, 89){\makebox(0,0)[b]{\smash{{\SetFigFont{10}{12.0}{\familydefault}{\mddefault}{\updefault}{\color[rgb]{0,0,0}$\bbisimd$}%
}}}}
\put(1651,-1486){\makebox(0,0)[b]{\smash{{\SetFigFont{10}{12.0}{\familydefault}{\mddefault}{\updefault}{\color[rgb]{0,0,0}(see Sect.~\ref{subsec:winr})}%
}}}}
\put(751,-511){\makebox(0,0)[rb]{\smash{{\SetFigFont{10}{12.0}{\familydefault}{\mddefault}{\updefault}{\color[rgb]{0,0,0}(see Sect.~\ref{subsec:closing})}%
}}}}
\end{picture}%

%% file: stuttering.pstex_t
\begin{picture}(0,0)%
\includegraphics{stuttering.pstex}%
\end{picture}%
\setlength{\unitlength}{3947sp}%
\begingroup\makeatletter\ifx\SetFigFont\undefined%
\gdef\SetFigFont#1#2#3#4#5{%
  \reset@font\fontsize{#1}{#2pt}%
  \fontfamily{#3}\fontseries{#4}\fontshape{#5}%
  \selectfont}%
\fi\endgroup%
\begin{picture}(1380,1111)(436,-425)
\put(1801,-361){\makebox(0,0)[b]{\smash{{\SetFigFont{10}{12.0}{\familydefault}{\mddefault}{\updefault}{\color[rgb]{0,0,0}$\astate{t}$}%
}}}}
\put(1801,539){\makebox(0,0)[b]{\smash{{\SetFigFont{10}{12.0}{\familydefault}{\mddefault}{\updefault}{\color[rgb]{0,0,0}$\astate{s}$}%
}}}}
\put(1051,539){\makebox(0,0)[b]{\smash{{\SetFigFont{10}{12.0}{\familydefault}{\mddefault}{\updefault}{\color[rgb]{0,0,0}$\states$}%
}}}}
\put(451,539){\makebox(0,0)[rb]{\smash{{\SetFigFont{10}{12.0}{\familydefault}{\mddefault}{\updefault}{\color[rgb]{0,0,0}$\bstate{s}$}%
}}}}
\put(451,-361){\makebox(0,0)[rb]{\smash{{\SetFigFont{10}{12.0}{\familydefault}{\mddefault}{\updefault}{\color[rgb]{0,0,0}$\bstate{t}$}%
}}}}
\put(1051,-361){\makebox(0,0)[b]{\smash{{\SetFigFont{10}{12.0}{\familydefault}{\mddefault}{\updefault}{\color[rgb]{0,0,0}$\statet$}%
}}}}
\end{picture}%

%% file: divergence.pstex_t
\begin{picture}(0,0)%
\includegraphics{divergence.pstex}%
\end{picture}%
\setlength{\unitlength}{3947sp}%
\begingroup\makeatletter\ifx\SetFigFont\undefined%
\gdef\SetFigFont#1#2#3#4#5{%
  \reset@font\fontsize{#1}{#2pt}%
  \fontfamily{#3}\fontseries{#4}\fontshape{#5}%
  \selectfont}%
\fi\endgroup%
\begin{picture}(3304,1636)(259,-875)
\put(1051,-61){\makebox(0,0)[b]{\smash{{\SetFigFont{10}{12.0}{\familydefault}{\mddefault}{\updefault}{\color[rgb]{0,0,0}$s_1$}%
}}}}
\put(1801,-61){\makebox(0,0)[b]{\smash{{\SetFigFont{10}{12.0}{\familydefault}{\mddefault}{\updefault}{\color[rgb]{0,0,0}$s_2$}%
}}}}
\put(2551,-61){\makebox(0,0)[b]{\smash{{\SetFigFont{10}{12.0}{\familydefault}{\mddefault}{\updefault}{\color[rgb]{0,0,0}$s_3$}%
}}}}
\put(1351, 14){\makebox(0,0)[b]{\smash{{\SetFigFont{10}{12.0}{\familydefault}{\mddefault}{\updefault}{\color[rgb]{0,0,0}$\tau$}%
}}}}
\put(2851, 14){\makebox(0,0)[b]{\smash{{\SetFigFont{10}{12.0}{\familydefault}{\mddefault}{\updefault}{\color[rgb]{0,0,0}$\tau$}%
}}}}
\put(2101, 14){\makebox(0,0)[b]{\smash{{\SetFigFont{10}{12.0}{\familydefault}{\mddefault}{\updefault}{\color[rgb]{0,0,0}$\tau$}%
}}}}
\put(601, 14){\makebox(0,0)[b]{\smash{{\SetFigFont{10}{12.0}{\familydefault}{\mddefault}{\updefault}{\color[rgb]{0,0,0}$\tau$}%
}}}}
\put(376,-61){\makebox(0,0)[rb]{\smash{{\SetFigFont{10}{12.0}{\familydefault}{\mddefault}{\updefault}{\color[rgb]{0,0,0}$s=s_0$}%
}}}}
\put(1126,239){\makebox(0,0)[lb]{\smash{{\SetFigFont{10}{12.0}{\familydefault}{\mddefault}{\updefault}{\color[rgb]{0,0,0}$\acta[2]$}%
}}}}
\put(1876,239){\makebox(0,0)[lb]{\smash{{\SetFigFont{10}{12.0}{\familydefault}{\mddefault}{\updefault}{\color[rgb]{0,0,0}$\acta[3]$}%
}}}}
\put(2626,239){\makebox(0,0)[lb]{\smash{{\SetFigFont{10}{12.0}{\familydefault}{\mddefault}{\updefault}{\color[rgb]{0,0,0}$\acta[4]$}%
}}}}
\put(376,239){\makebox(0,0)[lb]{\smash{{\SetFigFont{10}{12.0}{\familydefault}{\mddefault}{\updefault}{\color[rgb]{0,0,0}$\acta[1]$}%
}}}}
\put(376,-361){\makebox(0,0)[lb]{\smash{{\SetFigFont{10}{12.0}{\familydefault}{\mddefault}{\updefault}{\color[rgb]{0,0,0}$\actb[0]$}%
}}}}
\put(1126,-361){\makebox(0,0)[lb]{\smash{{\SetFigFont{10}{12.0}{\familydefault}{\mddefault}{\updefault}{\color[rgb]{0,0,0}$\actb[1]$}%
}}}}
\put(1876,-361){\makebox(0,0)[lb]{\smash{{\SetFigFont{10}{12.0}{\familydefault}{\mddefault}{\updefault}{\color[rgb]{0,0,0}$\actb[2]$}%
}}}}
\put(2626,-361){\makebox(0,0)[lb]{\smash{{\SetFigFont{10}{12.0}{\familydefault}{\mddefault}{\updefault}{\color[rgb]{0,0,0}$\actb[3]$}%
}}}}
\put(1051,614){\makebox(0,0)[b]{\smash{{\SetFigFont{10}{12.0}{\familydefault}{\mddefault}{\updefault}{\color[rgb]{0,0,0}$t_1$}%
}}}}
\put(1801,614){\makebox(0,0)[b]{\smash{{\SetFigFont{10}{12.0}{\familydefault}{\mddefault}{\updefault}{\color[rgb]{0,0,0}$t_2$}%
}}}}
\put(2551,614){\makebox(0,0)[b]{\smash{{\SetFigFont{10}{12.0}{\familydefault}{\mddefault}{\updefault}{\color[rgb]{0,0,0}$t_3$}%
}}}}
\put(301,614){\makebox(0,0)[b]{\smash{{\SetFigFont{10}{12.0}{\familydefault}{\mddefault}{\updefault}{\color[rgb]{0,0,0}$t_0$}%
}}}}
\put(301,-811){\makebox(0,0)[b]{\smash{{\SetFigFont{10}{12.0}{\familydefault}{\mddefault}{\updefault}{\color[rgb]{0,0,0}$u_0$}%
}}}}
\put(1051,-811){\makebox(0,0)[b]{\smash{{\SetFigFont{10}{12.0}{\familydefault}{\mddefault}{\updefault}{\color[rgb]{0,0,0}$u_1$}%
}}}}
\put(1801,-811){\makebox(0,0)[b]{\smash{{\SetFigFont{10}{12.0}{\familydefault}{\mddefault}{\updefault}{\color[rgb]{0,0,0}$u_2$}%
}}}}
\put(2551,-811){\makebox(0,0)[b]{\smash{{\SetFigFont{10}{12.0}{\familydefault}{\mddefault}{\updefault}{\color[rgb]{0,0,0}$u_3$}%
}}}}
\end{picture}%